\documentclass{article}

    \PassOptionsToPackage{numbers, compress}{natbib}

 \usepackage[preprint]{neurips_2026}

\usepackage[utf8]{inputenc}
\usepackage[T1]{fontenc}
\usepackage{hyperref}
\usepackage{url}
\usepackage{booktabs}
\usepackage{amsfonts}
\usepackage{nicefrac}
\usepackage{microtype}
\usepackage{amsmath} 
\usepackage{float}
\usepackage[linesnumbered,ruled,vlined]{algorithm2e}
\usepackage{graphicx}
\usepackage{booktabs}
\usepackage{multirow}
\usepackage{xcolor}
\usepackage{caption}
\definecolor{phasecolor}{RGB}{30,90,160}
\usepackage{wrapfig}
\usepackage{booktabs}
\usepackage{tabularx}
\usepackage{array}
\usepackage{makecell}
\newcolumntype{Y}{>{\raggedright\arraybackslash}X}

\newcommand{\phase}[1]{
    \vspace{2pt}
    {\color{phasecolor}\textbf{#1}}
    \vspace{2pt}
}

\usepackage{amsthm}

\newtheorem{theorem}{Theorem}[section]
\newtheorem{lemma}[theorem]{Lemma}
\newtheorem{proposition}[theorem]{Proposition}

\title{HACO: Hedged Agent Computing for Reliable LLM Systems}

\author{
  Enhan Li$^{1}$,   \quad
  Hongyang Du$^{1}$\thanks{Corresponding author.}
  \AND
  \normalfont
  $^{1}$The University of Hong Kong \quad
}

\begin{document}

\maketitle

\begin{abstract}
As large language model (LLM) agents move from isolated prompting to long-horizon workflows, failures increasingly arise at the role-to-instance binding boundary, where task-specific role requests must be assigned to concrete agent instances under current service, network, and query conditions.
Existing agent system research has improved role specialization, workflow topology, memory, and tool use, but often assumes a fixed stable execution environment. This assumption limits deployed reliability, because the same role request can exhibit different latency, failure probability, and output quality across agent instances operating under different service regions and network conditions.
We propose Hedged Agent Computing (HACO), a runtime control scheme that treats each role request as a reliability-constrained selection problem over candidate agent instances, each coupling a role type, an LLM, and a concrete execution environment. 
Different from routing, HACO adaptively selects a hedge set of candidates for each invocation. Its allocation rule combines optimistic ranking, which prioritizes candidates with high estimated quality, reliability, and informative uncertainty, with conservative reliability accumulation, which stops selection only after the hedge set reaches a target success probability. Through experience harvesting, HACO updates candidate and link profiles from all executed candidate traces, including quality, success, latency, and network statistics.
Experiments on various benchmarks, together with runtime degradation studies, show that HACO improves robustness and output quality under changing deployment conditions, while using lower token and latency cost than exhaustive parallel execution.
\end{abstract}

\section{Introduction}

\begin{wrapfigure}{r}{0.35\linewidth}
\vspace{-8pt}
\centering
    \includegraphics[width=0.4\textwidth]{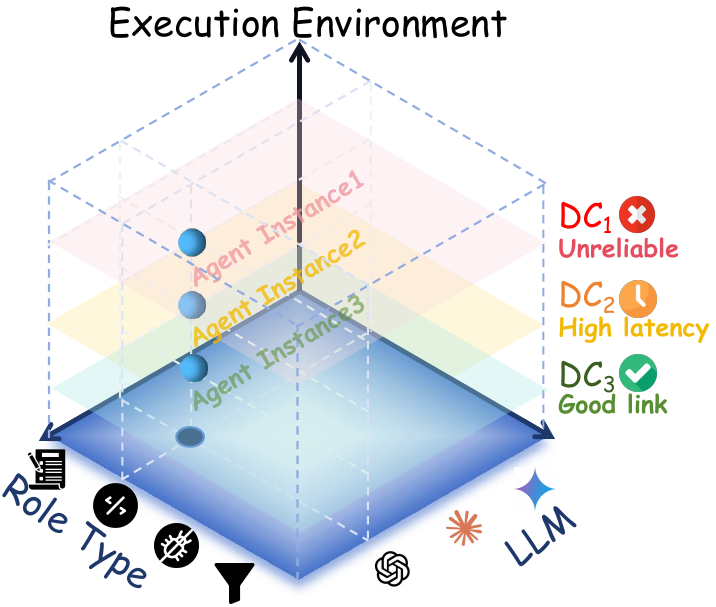}
    \caption{Three coupled axes of role invocation in MAS.}
    \label{fig:intro-coupled-axes}
\vspace{-10pt}
\end{wrapfigure}
LLM-based agent systems, especially multi-agent systems (MAS), have become a common paradigm for complex task execution and advanced automation~\citep{hong2023metagpt,qian2024chatdev,fourney2024magenticone,yan2025beyond}.
Compared with single-LLM execution, agent systems can distribute a task across specialized role types, coordinate intermediate outputs, cross-check partial results, and repair errors through iterative feedback. These coordination gains are especially valuable for tasks that require goal decomposition, intermediate coordination, iterative refinement, and long-horizon execution. 
As a result, agent systems have become an important direction for building practical AI systems, especially in domains such as data analysis, decision support, and scientific discovery.

This shift has also changed how practical agents are designed.
Modern agent systems are increasingly built with complete agent harnesses instead of standalone models with prompts \cite{anthropic2025harness,schmid2026agentharness}.
In practice, reliable execution depends on surrounding components such as tool interfaces, memory management, safety controls, and execution orchestration.
Recent work therefore places growing emphasis on persistent state~\cite{zhang2025survey}, controller design~\cite{mei2025survey}, execution layers, and tool coordination~\cite{qu2025tool}, because these components often determine whether agents can maintain progress through long and error-prone workflows.
However, this system perspective remains primarily organized around software-level MAS design, including role type specialization and system topology of agents.
When validating these designs, they often vary LLM backbones in the agent instances to show performance differences and robustness under the same MAS design~\cite{2025datawiseagent}.
As Fig.~\ref{fig:intro-coupled-axes} illustrates, once the system topology is fixed, this view attributes uncertainty mainly to the agent role type or the LLM's capability, such as reasoning or tool-use quality, while treating the execution environment of the agent instance as a fixed background condition.
The coupling among the capability of the selected LLM, its role type and the execution environment receives comparatively little attention.

However, the fixed-environment assumption breaks down in real deployments.
Even when the MAS software topology and role invocation remain unchanged, the final behavior depends heavily on the agent instance that performs the concrete operation. 
Specifically, a role invocation may pass through different service providers and execution regions. Its behavior is shaped by both LLM capability and physical environment factors, e.g., network latency, bandwidth fluctuation, transient failures, resource contention, and service interruptions, of the selected agent instance.
Consequently, two invocations that are identical at the software level may exhibit substantially different latency, execution reliability, and outcomes under different execution environments.
In practice, the role type is typically determined by the MAS design, while LLM choice and runtime conditions remain operational variables.
For example, modern LLM serving platforms expose region and deployment choices as runtime or configuration variables, so the same role invocation request may not always execute in the same physical location~\cite{azure_foundry_deployment_types,vertex_deployments_endpoints,bedrock_cross_region,openai_data_controls} (More examples are given in Appendix~\ref{sec:cross-region-gap}).
Moreover, recent infrastructure incidents, including Red Sea submarine-cable disruptions and Azure's West Europe thermal event, show that physical and regional execution conditions can directly affect service latency and availability~\cite{cnn2024redsea_cables,cnbc2025azure_redsea,azure2025west_europe_thermal}.
Therefore, agent systems should not attribute uncertainty only to model reasoning or software logic, but should also account for execution location, network paths, and region-level infrastructure state.

We therefore view role invocation as a coupled runtime decision over role type, LLM choice, and execution environment, and extend the conventional agent-instance abstraction by adding execution environment as a third axis. 
Inspired by the success of redundancy and fault-tolerance mechanisms in distributed systems, we use hedged concurrent execution to manage invocation-level uncertainty in computing and communication environments. 
HACO focuses on the role-invocation allocation point, which is different from single-route selection in \textit{LLM routing}~\citep{ding2024hybrid,ong2024routellm}, post-hoc escalation in \textit{fallback or cascade} methods~\citep{yue2023large,gupta2024language}, serving-batch construction in \textit{batch inference}~\citep{kwon2023vllm,zhong2024distserve}, and output selection or synthesis in \textit{MoA-style} or \textit{Best-of-$N$ aggregation}~\citep{jiang2023llmblender,wang2025mixture,wang2023selfconsistency}.
It asks a preceding execution-control question: before a role request is executed, whether multiple candidate agent instances should execute the same request, which candidates should be activated, and when the allocated redundancy is sufficient. 
Appendix~\ref{appcom} provides a structured comparison with related runtime and inference-time paradigms.
Motivated by these insights, we propose \textbf{Hedged Agent COmputing (HACO)}, where
the core selection unit is a \emph{candidate agent instance}, i.e., a concrete role-specific LLM deployment in an execution environment; we use \emph{candidate} as shorthand thereafter.
HACO allocates a small hedge set of candidates and uses the resulting execution traces to refine future allocation decisions.
Our contributions are summarized as:
\begin{itemize}
  \item We model role invocation in LLM-based agent systems as a three-axis runtime decision over role type, LLM choice, and execution environment. This formulation exposes the role-to-instance binding boundary as a source of reliability, latency, and quality variation under concrete runtime conditions.

\item We introduce HACO, a redundancy-based runtime control paradigm that selects reliability-constrained hedge sets over heterogeneous agent instances. Through experience harvesting, HACO refines candidate and link profiles, allowing the system to adapt redundancy to runtime conditions and balance robustness against execution cost.

\item We evaluate HACO across representative benchmarks and stress settings, showing that it improves robustness and output quality while offering a controllable trade-off among system reliability, task success, and execution cost measured by token usage and latency.
\end{itemize}

\section{Related Work and Theory Motivation}
\label{sec:related_work}
\paragraph{Agent harnesses.}
Recent work increasingly treats LLM agents as complete execution systems supported by agent harnesses, where the surrounding layer manages context construction, memory access, tool invocation, validation, and lifecycle control~\cite{anthropic2025harness,schmid2026agentharness,zhang2025survey,mei2025survey,qu2025tool}.
This harness view shows that practical agent performance depends on system mechanisms that maintain progress across long-horizon and failure-prone workflows, beyond the capability of the underlying model in a single turn. 
HACO builds on this system-level perspective, but studies a different control problem, i.e., how to allocate redundant execution when runtime conditions are uncertain.

\paragraph{LLM-based multi-agent systems.}
Research on LLM-based MAS mainly focuses on how agents are organized, how they communicate, and how they are controlled during complex tasks~\cite{he2025llm,yan2025beyond,raza2026trism}.
Representative systems include workflow-based role specialization such as MetaGPT~\cite{hong2023metagpt}, communication-centric collaboration such as ChatDev~\cite{qian2024chatdev}, and orchestrator-led delegation such as Magentic-One~\cite{fourney2024magenticone}.
These works improve collaboration structure, communication protocols, and controller design.
HACO addresses a complementary question, i.e., given a fixed MAS workflow, how should each role invocation be executed when candidate capability, network condition, and execution reliability vary at runtime.

\paragraph{Redundant and parallel agent execution.}
Redundancy is a common mechanism for improving reliability in deployed systems, but it also introduces execution cost~\cite{caict_alibaba_2025_genai_architecture}. 
Existing agent methods mainly use parallelism for full generation with post-hoc selection or structured wide search. For example, MoA-style method executes all candidates and selects the final output after generation~\cite{wang2025mixture}, while A-MapReduce and AgentSwing use structured parallelism for wide search and long-horizon branch routing~\cite{chen2026mapreduce,feng2026agentswing}. Agent-diversity studies further show that parallel agents can provide complementary information~\cite{yang2026understanding}. These methods demonstrate the value of parallel execution, but do not decide how much redundancy each role invocation needs under changing runtime conditions.
Redundancy is theoretically supported as a stabilizing mechanism under uncertainty and can reduce failure risk in multi-agent execution~\cite{bi2025redundancy}, but excessive redundancy increases token usage, latency, and coordination cost. HACO builds on this principle by treating redundancy as a query-dependent runtime control variable and formulating allocation as a reliability-constrained subset-selection problem:
\begin{equation}
\min_{\mathcal{S}\subseteq\mathcal{A}} \sum_{a_i\in\mathcal{S}} C_i
\quad
\text{s.t.}
\quad
\mathbb{P}(\text{qualified success}\mid \mathcal{S}) \ge \tau .
\end{equation}
Here, $\mathcal{S}$ is the hedge set, $\tau$ is the target reliability, and $C_i$ penalizes the candidate activation footprint, not realized wall-clock latency; the realized latency is $t(a_{\mathrm{winner}})$ in~\ref{sec:phase3}.

\section{The HACO Paradigm}
\label{sec:haco}

HACO treats redundancy as a runtime control variable, deciding how many and which candidate agent instances to activate for each role invocation. As shown in Fig.~\ref{fig:overview}, HACO consists of four stages, i.e., role abstraction, redundancy allocation, hedged execution, and experience harvesting.

\begin{figure}[t]
    \centering
    \includegraphics[width=1.0\textwidth]{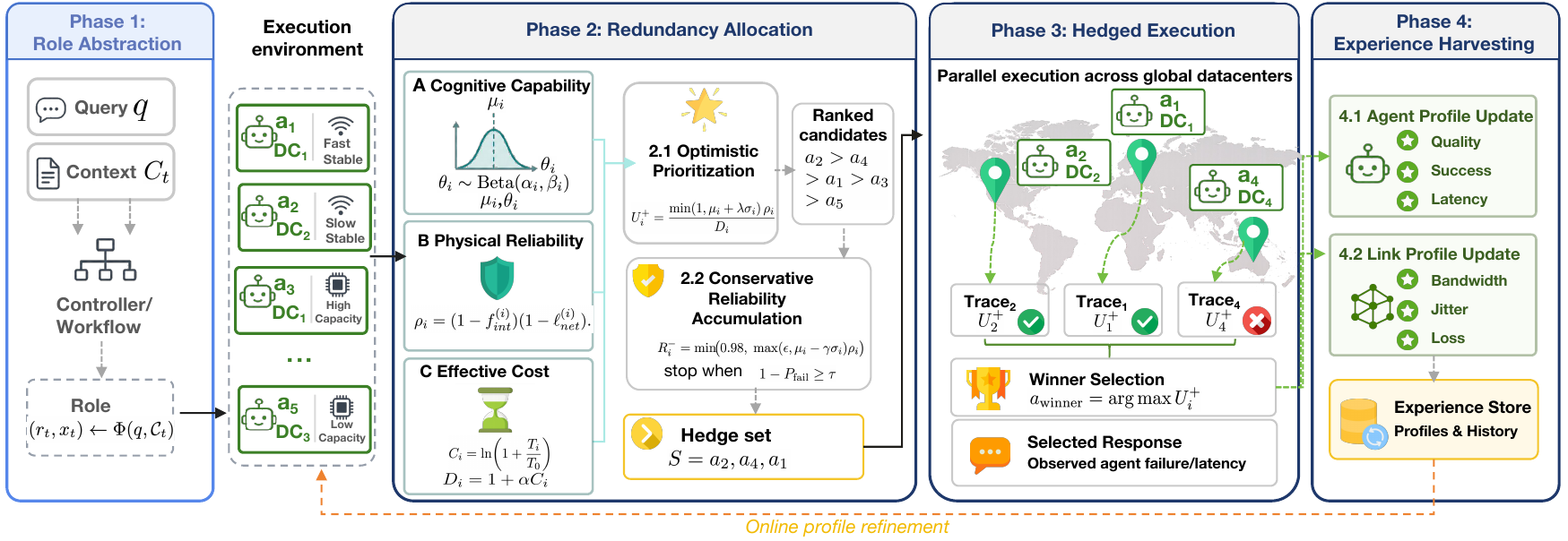}
    \caption{\textbf{Overview of the HACO paradigm.} 
    }  
    \label{fig:overview}
\end{figure}

\subsection{Phase 1: Role Abstraction}

HACO starts from role invocation events produced by the underlying MAS.
A role type refers to a functional component such as planner, coder, debugger, or filter.
At step $t$, a role invocation event is
\begin{equation}
e_t=(r_t,x_t),
\label{eq:role-invocation-event}
\end{equation}
where $r_t$ is the role type to invoke and $x_t$ is the input payload that is constructed from the original query $q$ and the current task context $\mathcal{C}_t$, including intermediate plans, execution results, error messages, and outputs from previous roles.
For example, a \textit{debugger} payload may contain the original query together with the failed \textit{coder} output, traceback, and current notebook state.

This abstraction connects the MAS controller to HACO's redundancy allocation mechanism.
Given Eq.~\eqref{eq:role-invocation-event}, HACO uses $r_t$ to identify the candidate pool and sends the same payload $x_t$ to the selected candidate agent instances.
The controller could be a predefined workflow, a finite-state machine, a decentralized negotiation protocol, a MetaGPT-style workflow~\cite{hong2023metagpt}, or a single-agent controller, depending on MAS design.
Without loss of generality, our experiments instantiate HACO with a relatively complex FSM-based controller derived from DATAWISE, as detailed in Appendix~\ref{sec: datawise_mas}. HACO itself only requires the controller to emit role invocation events in the form of Eq.~\eqref{eq:role-invocation-event}.

\subsection{Phase 2: Redundancy Allocation}
Given a role invocation, HACO selects a hedge set $\mathcal{S} \subseteq \mathcal{A}$ that meets a target reliability level with low latency-aware resource cost. To this end, HACO characterizes each candidate with uncertainty-aware utility components and allocates redundancy through dual-bound coordination.

\paragraph{Candidate Characterization.}
We characterize each candidate $a_i$ along three dimensions, i.e., uncertain cognitive capability, physical execution reliability, and latency-aware resource cost.

\emph{1. Cognitive Capability.}
Let $\theta_i \in [0,1]$ denote the expected normalized output-quality capability of candidate $a_i$. HACO models $\theta_i$ using a Beta prior, a standard choice for variables restricted to the unit interval~\cite{ferrari2004beta, kull2017beta, ma2024beta}:
\begin{equation}
\theta_i \sim \mathrm{Beta}(\alpha_i,\beta_i),
\label{eq:capability-prior}
\end{equation}
which is initialized with $\alpha_i=\beta_i=1$. After observing a normalized output-quality score $q_i \in [0,1]$, HACO applies a lightweight soft-count update:
\begin{equation}
\alpha_i \leftarrow \alpha_i + q_i,
\qquad
\beta_i \leftarrow \beta_i + (1-q_i).
\label{eq:capability-update}
\end{equation}
The posterior mean and standard deviation are
\begin{equation}
\mu_i=\frac{\alpha_i}{\alpha_i+\beta_i},
\qquad
\sigma_i=
\sqrt{
\frac{\alpha_i\beta_i}
{(\alpha_i+\beta_i)^2(\alpha_i+\beta_i+1)}
}.
\label{eq:posterior-stats}
\end{equation}
Here, $\mu_i$ estimates the expected capability of the selected candidate $a_i$, while $\sigma_i$ captures uncertainty.

\emph{2. Physical Execution Reliability.}
Let $f_{\mathrm{int}}^{(i)}$ denote the internal failure rate of candidate $a_i$, and let $\ell_{\mathrm{net}}^{(i)}$ denote the network loss rate. We define the execution success probability as
\begin{equation}
\rho_i = (1-f_{\mathrm{int}}^{(i)})(1-\ell_{\mathrm{net}}^{(i)}).
\label{eq:physical-reliability}
\end{equation}
For tractability, HACO considers conditional independence across candidates when estimating Eq.~\ref{eq:physical-reliability}.

\emph{3. Latency-Aware Resource Cost.}
For each candidate, HACO estimates an effective latency by combining the candidate's processing time and the expected communication overhead:
\begin{equation}
T_i=
t_{\mathrm{proc}}^{(i)}+
\frac{k_i}{\max(\epsilon_{\rm B}, b_i-\kappa j_i)}.
\label{eq:effective-latency}
\end{equation}
Here, $t_{\mathrm{proc}}^{(i)}$ denotes the historical average processing latency of candidate $i$, $k_i$ is the estimated message size, $b_i$ is the observed effective bandwidth, and $j_i$ is the unit bandwidth jitter.
The coefficient $\kappa>0$ models how strongly bandwidth jitter reduces the effective bandwidth term.

To reflect the diminishing sensitivity of user-perceived delay, HACO defines a scale-normalized logarithmic latency cost with the estimated latency \(T_i\) based on the Weber--Fechner law~\cite{reichl2010logarithmic}:
\begin{equation}
C_i=\ln\!\left(1+\frac{T_i}{T_0}\right),
\label{eq:latency-cost}
\end{equation}
where $T_0$ is a reference latency, instantiated as the median estimated latency among candidates for the current routing step. HACO converts this latency-aware resource cost into a stable discount factor:
\begin{equation}
D_i = 1 + \eta C_i,
\label{eq:delay-discount}
\end{equation}
where $\eta$ controls the strength of latency penalization. Compared with directly using $C_i$, this form preserves the logarithmic delay effect, penalizes larger latency monotonically, and avoids excessive punishment of high-latency candidates and excessive reward for near-zero latency.

\paragraph{Dual-Bound Coordination.}

Based on the candidate characterization results, HACO allocates redundancy with a dual-bound coordination mechanism.
Candidate ranking and redundancy termination face different uncertainty risks.
For ranking, overly pessimistic estimates may hide under-observed but potentially strong candidates.
For termination, overly optimistic estimates may stop the hedge construction too early and weaken system reliability.
HACO addresses these two risks with separate estimates, i.e., an optimistic utility for ranking candidates and a conservative reliability estimate for deciding when the hedge set is sufficient.

\emph{1. Optimistic Prioritization.}
For each candidate $a_i$, HACO computes an optimism-adjusted utility
\begin{equation}
U_i^{+}=
\frac{
\min(1,\mu_i+\lambda\sigma_i)\,\rho_i
}{
D_i
},
\label{eq:optimistic-utility}
\end{equation}
where $\mu_i$ and $\sigma_i$ are the mean and standard deviation of the candidate's posterior output-quality estimation (i.e., Eq.~\ref{eq:posterior-stats}) , $\lambda$ is the exploration coefficient, $\rho_i$ denotes the estimated execution success probability (i.e., Eq.~\ref{eq:physical-reliability}), and $D_i$ is the stable latency discount factor (i.e., Eq.~\ref{eq:delay-discount}). Candidates are ranked in descending order of $U_i^{+}$.
This utility favors candidates with high estimated capability of the selected LLM and high execution reliability, while applying a controlled latency discount. The uncertainty bonus $\lambda\sigma_i$ encourages exploration of under-observed but potentially strong candidates. 

\emph{2. Conservative Reliability Accumulation.}
Following the ranking, HACO incrementally constructs the hedge set using a pessimistic execution reliability estimate:
\begin{equation}
R_i^- = \min\left(0.98, \max(\epsilon_{\rm R}, \mu_i-\gamma\sigma_i)\cdot \rho_i\right).
\label{eq:conservative-reliability}
\end{equation}
Here, $\mu_i-\gamma\sigma_i$ gives a lower-confidence estimate of the candidate's output-quality capability, and $\rho_i$ discounts this estimate by the candidate's physical execution reliability.
The lower bound $\epsilon_{\rm R}$ stabilizes the capability term, while the upper cap $0.98$ avoids perfect single-candidate reliability.
This conservative estimate is used only for stopping, so that HACO does not terminate hedge construction based on overly optimistic reliability estimates.

Let $P_{\mathrm{fail}}$ denote the cumulative failure probability that all currently selected candidates fail. Starting from $P_{\mathrm{fail}}=1$, HACO updates:
\begin{equation}
P_{\mathrm{fail}}
\leftarrow
P_{\mathrm{fail}}(1-R_i^{-})
\label{eq:failure-update}
\end{equation}
whenever a new candidate is added. Redundancy allocation stops when
\begin{equation}
1-P_{\mathrm{fail}} \ge \tau,
\end{equation}
or when the candidate pool is exhausted.
A single shared estimate is insufficient for both purposes: an overly conservative utility estimate would suppress exploration, whereas an overly optimistic one would weaken the stopping guarantee.

In this way, Phase 2 decides both \emph{which} candidates should enter the hedge set and \emph{how many} to select under the conservative stopping rule.
The resulting ordering-and-stopping rule is analyzed in Appendix~\ref{app:haco_properties}, which establishes conservative reliability under the $U_i^+$-induced ordering.

\subsection{Phase 3: Hedged Execution}
\label{sec:phase3}
Given the hedge set selected, HACO launches all candidates in $\mathcal{S}$ in parallel and returns once the highest-utility successful candidate is determined.
Remaining traces are harvested asynchronously for future profile updates.
For each executed candidate $a_i \in \mathcal{S}$, let $s_i \in \{0,1\}$ denote whether the execution succeeds, and let $t_i$ denote its observed completion time.
HACO denotes the execution-successful candidates observed for winner selection as
\begin{equation}
\mathcal{S}^{+} = \{ a_i \in \mathcal{S} \mid s_i = 1 \}.
\label{eq:successful-set}
\end{equation}
If $\mathcal{S}^+=\emptyset$, HACO emits a role-failure payload to the controller for recovery and skips winner selection.
Otherwise, HACO returns the candidate with the highest optimistic prioritization utility:
\begin{equation}
a_{\mathrm{winner}}
=
\arg\max_{a_i \in \mathcal{S}^{+}} U_i^{+}.
\label{eq:winner-selection}
\end{equation}
The observed invocation latency is the time when the selected winner can be determined:
\begin{equation}
T_{\mathrm{obs}}
=
\max\!\left(
t(a_{\mathrm{winner}}),
\max_{a_j \in \mathcal{S}: U_j^+ > U_{a_{\mathrm{winner}}}^+}
t_j^{\mathrm{res}}
\right),
\label{eq:observed-latency}
\end{equation}
where $t_j^{\mathrm{res}}$ denotes the resolution time of a higher-utility candidate, and the inner maximum is omitted when the set is empty.
Thus, HACO ties blocking latency to the winner-decision time and uses the same \(U_i^+\) signal for final selection, trading post-execution reranking for lower evaluation overhead.

\subsection{Phase 4: Experience Harvesting}
\label{sec:experience_harvesting}

A key advantage of HACO is that redundancy is used not only for robust execution, but also for experience harvesting since all executed candidates in the hedge set could provide useful feedback for future allocation decisions.

Specifically, for each executed candidate $a_i \in \mathcal{S}_t$, HACO records its LLM-judge quality signal $q_i$, success status $s_i$, completion time $t_i$, token usage, internal failure status, and communication telemetry. Here, $\mathcal{S}_t$ denotes the hedge set at step $t$ (and we write $\mathcal{S}$ when the step index is omitted). These traces are aggregated into two complementary profiles:

\textit{Candidate profiles.} HACO updates uncertainty-aware capability estimate from LLM-judge quality signals. Here, $q_i$ is used to update the Beta profile in Eq.~\eqref{eq:capability-update} and changes the posterior mean and uncertainty $(\mu_i,\sigma_i)$ in Eq.~\eqref{eq:posterior-stats}. These updated estimates are reused in the next allocation step by the optimistic ranking score $U_i^+$ in Eq.~\eqref{eq:optimistic-utility} and the conservative reliability estimate $R_i^-$ in Eq.~\eqref{eq:conservative-reliability}. HACO also refreshes running statistics, e.g., call frequency, failure rate, latency, and token usage.

\textit{Network-link profiles.} HACO logs inter-agent transmissions and updates link-level statistics, including average bandwidth, bandwidth jitter, and packet-loss rate. These statistics are reused to estimate physical reliability $\rho_i$ in Eq.~\eqref{eq:physical-reliability} and effective latency $T_i$ in Eq.~\eqref{eq:effective-latency}. The resulting latency cost and delay discount in Eqs.~\eqref{eq:latency-cost}--\eqref{eq:delay-discount} then affect the next-round utility score and hedge-set construction.

This harvesting mechanism turns redundant execution into both immediate robustness and future allocation evidence. Each hedge set provides observations over agent capability and physical execution conditions, allowing HACO to refine agent selection and network-aware cost estimation over time. Algorithm~\ref{alg:haco} in Appendix~\ref{sec:algorithm} summarizes the full procedure of HACO.

\section{Experiments}

\subsection{Experimental Setting}

\noindent \textbf{Role types.}
HACO can be applied to any MAS that exposes role invocations and allows execution over multiple candidate agent instances. Here, we instantiate HACO on a multi-agent data-analysis system derived from \textsc{DATAWISE}. 
The underlying MAS is organized around four role types: planner, coder, debugger, and filter. 
These roles form a structured workflow in which the planner decomposes and updates task steps, the coder generates and executes code, the debugger repairs execution failures, and the filter cleans debugging outputs before execution resumes. 
This setting provides a representative long-horizon agent controller with state transitions and self-debugging behavior, while leaving HACO independent of the specific controller implementation. 
The full MAS-DATAWISE workflow is described in Appendix~\ref{sec: datawise_mas}.

\noindent \textbf{Execution environment.}
To evaluate runtime uncertainty, we build a heterogeneous execution environment that captures both network fluctuation and probabilistic candidate failures. 
Candidate agent instances are deployed across three execution zones, i.e., Global, Regional, and Local. 
Each zone contains LLM candidates instantiated through different model APIs. We configure candidate-level invocation failures and inter-zone network conditions, including latency, jitter, bandwidth, and loss rate, through a controllable interface that emulates realistic deployment conditions.
The motivation for using execution zones follows the cross-region inference patterns discussed in Appendix~\ref{sec:cross-region-gap}; concrete candidate-pool configurations are reported in Appendix~\ref{app:candidate_pool_details}, and the network simulator is detailed in Appendix~\ref{app:network_simulation}. 
We also validate HACO in a real-world Azure deployment with regional endpoints, measured endpoint latency, and candidate failures, as detailed in Appendix~\ref{app:realworld_setting}.

\noindent \textbf{Benchmarks.}
We evaluate HACO on three representative benchmarks. 
DSBench~\cite{jing2025dsbench} evaluates data science agents on end-to-end workflows involving data preprocessing, feature engineering, and predictive modeling.
InfiAgent-Bench~\cite{hu2024infiagent} focuses on multi-step data analysis over structured CSV files in a code execution environment.
MatplotBench~\cite{yang2024matplotagent} evaluates scientific visualization agents that generate plots from natural language specifications and structured data inputs.
These benchmarks cover data-analysis reasoning, executable workflows, and visualization-oriented code generation.

\noindent \textbf{Baselines.}
We compare HACO with representative baselines.
Following recent LLM-routing benchmark practice~\cite{huang2025routereval}, \textbf{Random} uniformly samples one candidate for execution, serving as an environment-agnostic single-route reference.
Inspired by recent feedback-based LLM routing methods~\cite{ong2024routellm,tsiourvas2025causal}, \textbf{BestOne} selects the best one candidate according to historical performance records.
\textbf{MoA-style Best-of-\(N\)}~\cite{wang2025mixture} executes all candidate agents in parallel and uses an evaluator language model to score the candidate outputs, returning the highest-scoring response.
Detailed descriptions and pseudocode are provided in Appendix~\ref{app:baseline}; hyperparameters are reported in Appendix~\ref{sec:algorithm}.

\subsection{Experimental Results}
\label{sec:experiments}

\paragraph{Reliability under Runtime Uncertainty.}
Fig.~\ref{fig:robustness} evaluates HACO under runtime degradation injected after task index 50 and under an Azure-backed bad-agent deployment.
Across network degradation, candidate degradation, and real bad-agent deployment, BestOne shows the largest reliability collapse, with all-failed-rate increases of $72.8$, $67.4$, and $90.4$ percentage points.
The reason is that BestOne repeatedly favors the historically strongest candidate or zone, which becomes brittle once that execution path is degraded.
Random is less concentrated, but it still lacks an explicit mechanism to identify and avoid degraded runtime conditions.
HACO remains stable because redundant executions provide online evidence for reallocating future invocations.
The selected-zone shifts in Fig.~\ref{fig:robustness}(d) show that HACO moves away from degraded regions, reducing Global selections by $80$ and $40$ percentage points under network and candidate degradation, respectively.
In the real deployment, HACO shifts away from the failing US-labeled deployment and increases the use of JP and SE, while BestOne shifts further toward the failing region.
Evaluator degradation has a minor effect, indicating that HACO does not strongly depend on the judge and can operate effectively with a weak LLM-as-judge evaluator.
Overall, HACO approaches MoA-style robustness without full-pool execution, supporting the value of experience-guided adaptive redundancy under runtime uncertainty.
\begin{figure}[t]
    \centering
    \includegraphics[width=1.0\textwidth]{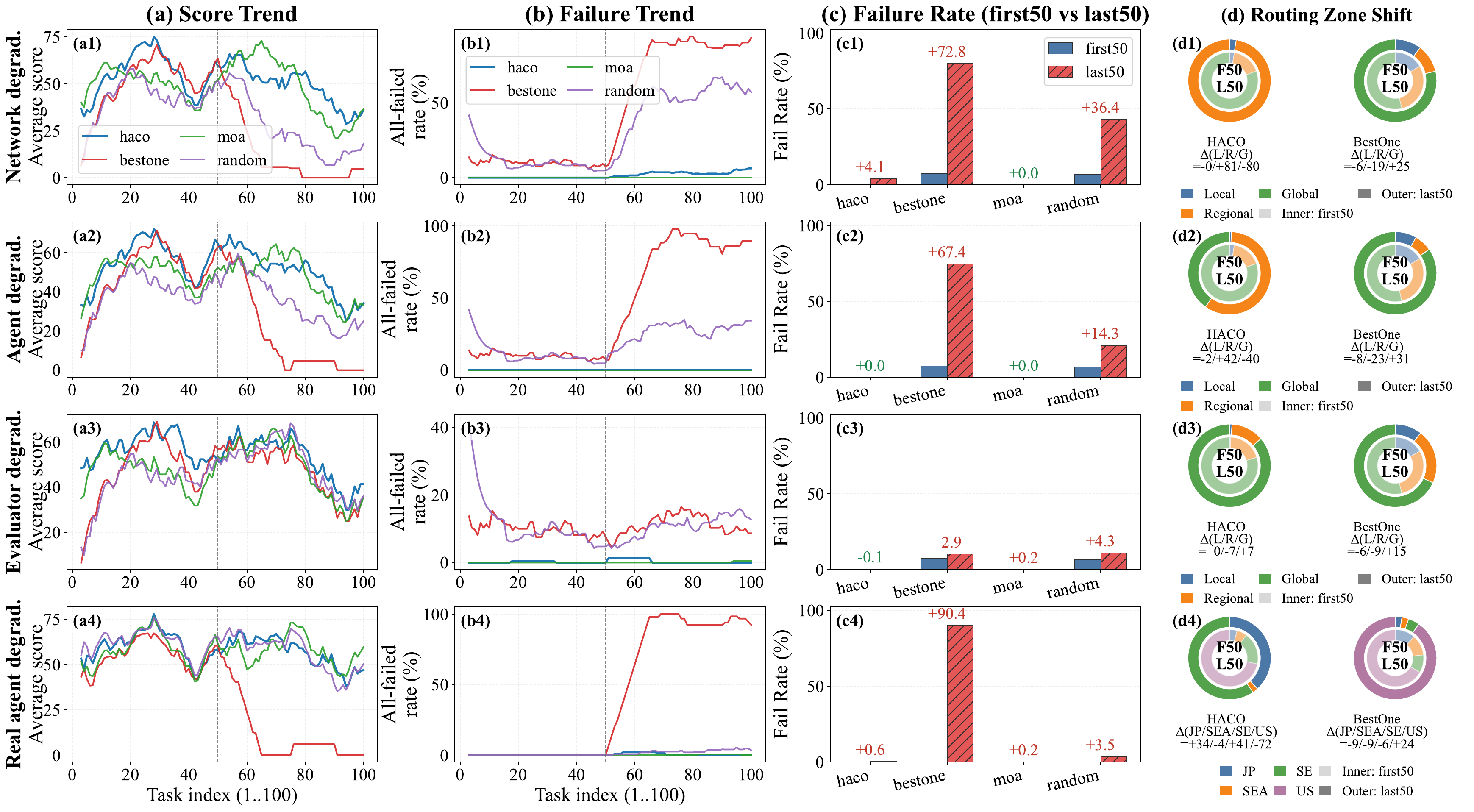}
    \caption{\textbf{Reliability evaluation under runtime degradation and real deployment failures.} Rows correspond to network degradation, candidate degradation, evaluator degradation, and a real-world Azure-backed deployment with injected candidate failures.}
    \label{fig:robustness}
\end{figure}

\paragraph{Overall Benchmark Performance.}
\label{sec:general_performance}

Fig.~\ref{fig:radar} compares HACO with single-route, fixed-zone, and full-redundancy baselines under the heterogeneous execution environment.
All axes are normalized so that larger values indicate better performance, including effectiveness, strict success, token efficiency, latency efficiency, and all-failed-step reliability.
Formal definitions of these radar metrics are provided in Appendix~\ref{app:radar_metric_defs}.
Across all benchmarks, HACO forms one of the largest and most balanced radar profiles, indicating consistent gains in task effectiveness, execution reliability, and practical efficiency.
Single-route baselines expose the limitation of committing to one candidate before execution: Random is inexpensive but unstable, while BestOne can exploit historical performance but remains sensitive to unreliable candidates or execution zones.
Fixed-zone baselines further show that no single execution zone dominates across all benchmarks and metrics.
MoA-style Best-of-N provides strong reliability through full-pool execution, but this comes with substantially higher token and latency cost.
In contrast, HACO uses an adaptive hedge subset, achieving reliability close to MoA while avoiding exhaustive execution.
Detailed numerical results and an external-reference estimate are reported in Appendix~\ref{app:main}.
Overall, these results support the main claim that adaptive redundancy provides a better effectiveness--efficiency--reliability trade-off than either single routing or full parallel execution.

\begin{figure}[t]
    \centering
    \includegraphics[width=0.95\textwidth]{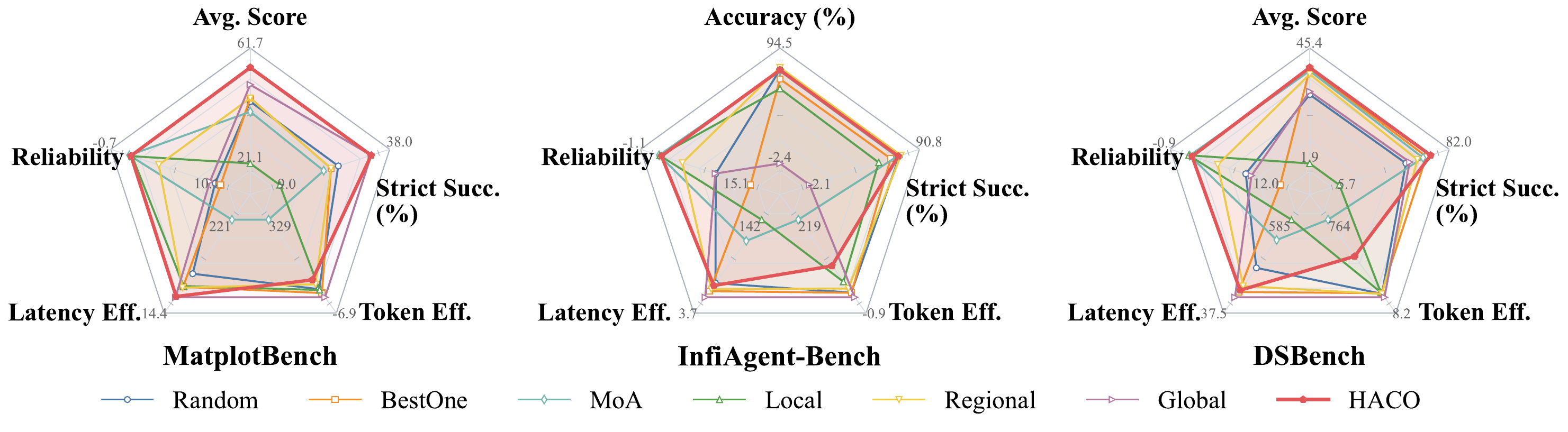}
    \caption{\textbf{Performance comparison under the heterogeneous execution environment.} Results are reported on MatplotBench, InfiAgent-Bench, and DSBench for Random, BestOne, MoA-style Best-of-N, fixed-zone Local, fixed-zone Regional, fixed-zone Global, and HACO.}
    \label{fig:radar}
\end{figure}

\paragraph{Behavioral Analysis.}
\label{sec:behavior}

Fig.~\ref{fig:behavior} examines HACO's adaptive redundancy and routing behavior.
Unlike Random, BestOne, and fixed-zone baselines, which execute one candidate per invocation, and unlike MoA-style baseline, which executes the full candidate pool, HACO selects a hedge set whose size varies adaptively to adjust redundancy according to task difficulty and runtime uncertainty, as shown in Fig.~\ref{fig:behavior} (a).
The per-task traces in Fig.~\ref{fig:behavior} (b) show role- and dataset-dependent allocation patterns, with harder workflows such as DSBench generally requiring larger hedge sets than simpler analysis tasks such as InfiAgent-Bench.
The selected-zone distribution in Fig.~\ref{fig:behavior} (c) shows that HACO uses multiple execution zones instead of committing to a single fixed zone.
These results support HACO's design goal of treating redundancy and execution environment as runtime decision variables.

\begin{figure}[t]
    \centering
    \includegraphics[width=1\textwidth]{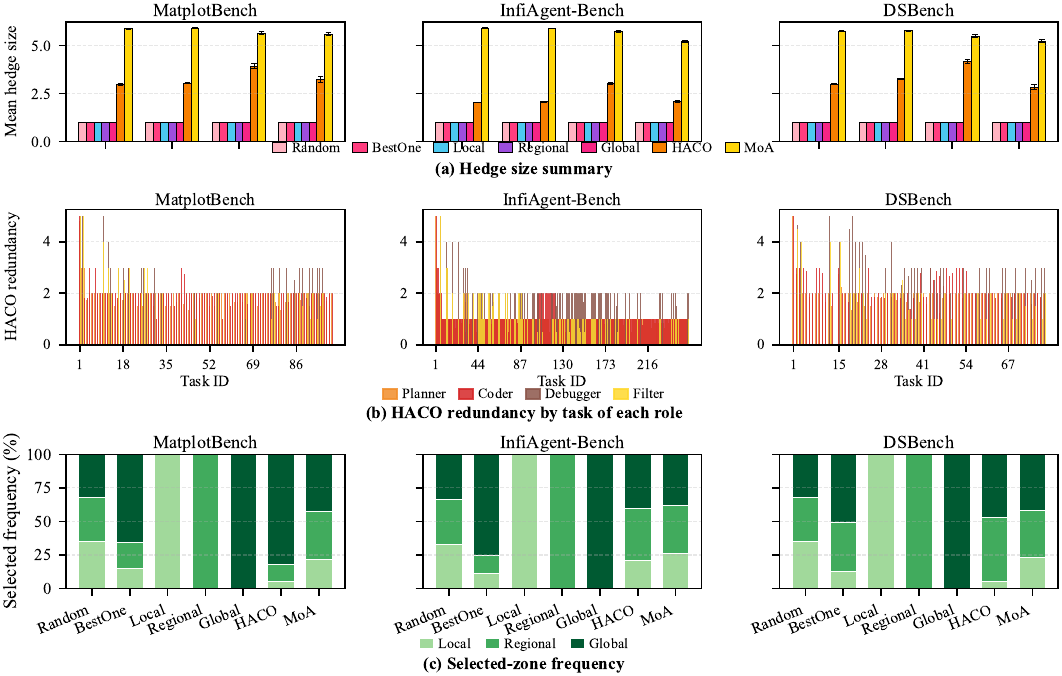}

    \caption{\textbf{HACO behavior.}
    (a) Mean hedge size of baselines and HACO, (b) per-task role redundancy, and (c) selected-zone frequency across three benchmarks.}
    \label{fig:behavior}
\end{figure}

\paragraph{Reliability–Cost Trade-off.}
\label{sec:tradeoff}

Fig.~\ref{fig:tradeoff} (a)--(c) show that the reliability target $\tau$ provides a controllable knob for balancing task quality, execution reliability, and token cost.
As $\tau$ increases from $0.70$ to $0.95$, the step failure rate decreases from $1.87\%$ to $0.17\%$, while the average score increases from $53.77$ to $58.86$.
This improvement comes at the cost of more redundant execution, as reflected by the increasing token mean in Fig.~\ref{fig:tradeoff} (c).
The score-token curve changes smoothly, indicating that HACO can move gradually between low-cost and high-reliability operating points.
Single-route methods operate in a low-cost but less reliable regime, while MoA improves reliability through full-pool execution.
HACO occupies an intermediate region, approaching high reliability without always paying the full cost of exhaustive parallelism.
These results support the claim that adaptive redundancy improves robustness and exposes a tunable performance-efficiency frontier.

\begin{figure}[t]
    \centering
    \includegraphics[width=1.0\textwidth]{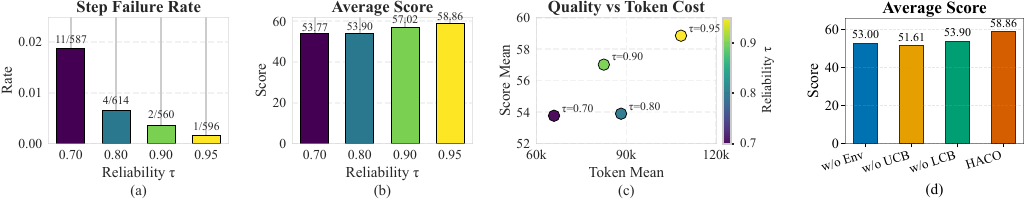}
    \caption{\textbf{Reliability-target sweep and component ablation in MatplotBench.} Panels (a)--(c) sweep the system reliability target $\tau \in \{0.70, 0.80, 0.90, 0.95\}$ and report step failure rate, average score, and the score-token relationship. Panel (d) compares HACO with ablated variants that remove environment-aware modeling ({\rm{Env}}), the uncertainty bonus ({\rm{UCB}}), or conservative stopping ({\rm{LCB}}).}
    \label{fig:tradeoff}
\end{figure}

\paragraph{Component Ablation}
\label{sec:ablation}

Fig.~\ref{fig:tradeoff} (d) evaluates the contribution of three key components in HACO on MatplotBench.
Removing any component decreases the average score, indicating that HACO's gain does not come from a single heuristic.
Without environment-aware modeling, the score drops from $58.86$ to $53.00$, showing that effective candidate selection requires communication reliability, latency, and output quality to be modeled jointly.
Removing the uncertainty bonus gives the lowest score, $51.61$, suggesting that optimistic exploration is important for discovering under-observed but useful candidates.
Replacing conservative stopping also reduces the score to $53.90$, which suggests that pessimistic reliability estimation helps avoid premature termination of the hedge set.
Overall, the ablation confirms that HACO benefits from the joint design of environment-aware routing, uncertainty-aware ranking, and conservative reliability accumulation.

\section{Conclusion}

We presented HACO, a hedged agent computing approach for LLM-based agent systems.
HACO views each role invocation as a runtime decision that jointly depends on role type, LLM choice, and execution environment, instead of treating it as a fixed call determined only by a role--agent instance pair.
It combines query-adaptive redundant execution with experience harvesting, so that redundant invocations improve current robustness while also providing feedback for future allocation.
Across agent benchmarks and stress settings, HACO improves execution reliability and output quality under heterogeneous runtime conditions, including network degradation, candidate degradation, evaluator degradation, and an Azure-backed deployment with injected candidate failures. HACO reduces the fragility of single-route execution and approaches the robustness of full parallel execution without always incurring its full token and latency cost. 
These results suggest that execution-aware redundancy allocation is a practical mechanism for building reliable LLM-based agent systems in dynamic real-world deployment environments.

\bibliographystyle{plainnat}
\bibliography{neurips_2026}

@misc{anthropic2025harness,
  title={Effective Harnesses for Long-Running Agents},
  author={{Anthropic}},
  year={2025},
  howpublished={Anthropic Engineering Blog},
  note={Published Nov. 26, 2025},
  url={https://www.anthropic.com/engineering/effective-harnesses-for-long-running-agents}
}

@misc{schmid2026agentharness,
  title={The Importance of Agent Harness in 2026},
  author={Schmid, Philipp},
  year={2026},
  howpublished={Blog post},
  url={https://www.philschmid.de/agent-harness-2026}
}

@article{zhang2025survey,
  title={A survey on the memory mechanism of large language model-based agents},
  author={Zhang, Zeyu and Dai, Quanyu and Bo, Xiaohe and Ma, Chen and Li, Rui and Chen, Xu and Zhu, Jieming and Dong, Zhenhua and Wen, Ji-Rong},
  journal={ACM Transactions on Information Systems},
  volume={43},
  number={6},
  pages={1--47},
  year={2025},
  publisher={ACM New York, NY}
}

@article{mei2025survey,
  title={A survey of context engineering for large language models},
  author={Mei, Lingrui and Yao, Jiayu and Ge, Yuyao and Wang, Yiwei and Bi, Baolong and Cai, Yujun and Liu, Jiazhi and Li, Mingyu and Li, Zhong-Zhi and Zhang, Duzhen and others},
  journal={arXiv preprint arXiv:2507.13334},
  year={2025}
}

@article{qu2025tool,
  title={Tool learning with large language models: A survey},
  author={Qu, Changle and Dai, Sunhao and Wei, Xiaochi and Cai, Hengyi and Wang, Shuaiqiang and Yin, Dawei and Xu, Jun and Wen, Ji-Rong},
  journal={Frontiers of Computer Science},
  volume={19},
  number={8},
  pages={198343},
  year={2025},
  publisher={Springer}
}

@article{he2025llm,
  title={{LLM}-based multi-agent systems for software engineering: Literature review, vision, and the road ahead},
  author={He, Junda and Treude, Christoph and Lo, David},
  journal={ACM Transactions on Software Engineering and Methodology},
  volume={34},
  number={5},
  pages={1--30},
  year={2025},
  publisher={ACM New York, NY}
}

@article{yan2025beyond,
  title={Beyond self-talk: A communication-centric survey of {LLM}-based multi-agent systems},
  author={Yan, Bingyu and Zhou, Zhibo and Zhang, Litian and Zhang, Lian and Zhou, Ziyi and Miao, Dezhuang and Li, Zhoujun and Li, Chaozhuo and Zhang, Xiaoming},
  journal={arXiv preprint arXiv:2502.14321},
  year={2025}
}

@article{raza2026trism,
  title={{Trism} for agentic {Trism}: A review of trust, risk, and security management in {LLM}-based agentic multi-agent systems},
  author={Raza, Shaina and Sapkota, Ranjan and Karkee, Manoj and Emmanouilidis, Christos},
  journal={AI Open},
  year={2026},
  publisher={Elsevier}
}

@article{qian2024chatdev,
  title={ChatDev: Communicative Agents for Software Development},
  author={Qian, Chen and Liu, Wei and Liu, Hongzhang and Chen, Nuo and Dang, Yufan and Li, Jiahao and Yang, Cheng and Chen, Weize and Su, Yusheng and Cong, Xin and Xu, Juyuan and Li, Dahai and Liu, Zhiyuan and Sun, Maosong},
  journal={ACL},
  year={2024},
  url={https://arxiv.org/abs/2307.07924}
}

@article{fourney2024magenticone,
  title={Magentic-One: A Generalist Multi-Agent System for Solving Complex Tasks},
  author={Fourney, Adam and Bansal, Gagan and Mozannar, Hussein and Tan, Cheng and Salinas, Eduardo and Zhu, Erkang and Niedtner, Friederike and Proebsting, Grace and Bassman, Griffin and Gerrits, Jack and Alber, Jacob and Chang, Peter and Loynd, Ricky and West, Robert and Dibia, Victor and Awadallah, Ahmed and Kamar, Ece and Hosn, Rafah and Amershi, Saleema},
  journal={arXiv preprint arXiv:2411.04468},
  year={2024},
  url={https://arxiv.org/abs/2411.04468}
}

@misc{caict_alibaba_2025_genai_architecture,
  author       = {{China Academy of Information and Communications Technology} and {Alibaba Cloud Computing Co., Ltd.}},
  title        = {Guiding Principles for Excellent Generative AI Architecture Design},
  year         = {2025},
  month        = sep,
  day          = {17},
  url          = {https://pdf.dfcfw.com/pdf/H3_AP202509181745854315_1.pdf},
  note         = {Accessed: 2026-04-07}
}

@article{feng2026agentswing,
  title={AgentSwing: Adaptive Parallel Context Management Routing for Long-Horizon Web Agents},
  author={Feng, Zhaopeng and Su, Liangcai and Zhang, Zhen and Wang, Xinyu and Zhang, Xiaotian and Wang, Xiaobin and Fang, Runnan and Zhang, Qi and Li, Baixuan and Cai, Shihao and others},
  journal={arXiv preprint arXiv:2603.27490},
  year={2026}
}

@article{chen2026mapreduce,
  title={A-MapReduce: Executing Wide Search via Agentic MapReduce},
  author={Chen, Mingju and Zhang, Guibin and Chang, Heng and Guo, Yuchen and Zhou, Shiji},
  journal={arXiv preprint arXiv:2602.01331},
  year={2026}
}

@article{yang2026understanding,
  title={Understanding Agent Scaling in LLM-Based Multi-Agent Systems via Diversity},
  author={Yang, Yingxuan and Qu, Chengrui and Wen, Muning and Shi, Laixi and Wen, Ying and Zhang, Weinan and Wierman, Adam and Gu, Shangding},
  journal={arXiv preprint arXiv:2602.03794},
  year={2026},
  url={https://arxiv.org/abs/2602.03794}
}

@misc{cnn2024redsea_cables,
  title        = {Red Sea Cables Have Been Damaged, Disrupting Internet Traffic},
  author       = {{CNN Business}},
  year         = {2024},
  howpublished = {\url{https://www.cnn.com/2024/03/04/business/red-sea-cables-cut-internet}},
  note         = {Accessed: 2026-05-02}
}

@misc{cnbc2025azure_redsea,
  title        = {Microsoft Says Azure Cloud Computing Service Disrupted by Fiber Cuts in the Red Sea},
  author       = {{CNBC}},
  year         = {2025},
  howpublished = {\url{https://www.cnbc.com/2025/09/06/microsoft-azure-cloud-computing-service-disrupted-red-sea-fiber-cuts.html}},
  note         = {Accessed: 2026-05-02}
}

@misc{azure2025west_europe_thermal,
  title        = {Post Incident Review: Thermal Event Impacting Multiple Services -- West Europe},
  author       = {{Microsoft Azure}},
  year         = {2025},
  howpublished = {\url{https://azure.status.microsoft/en-us/status/history/?force_isolation=true}},
  note         = {Tracking ID: 2LGD-9VG. Accessed: 2026-05-02}
}

@article{ferrari2004beta,
  title={Beta regression for modelling rates and proportions},
  author={Ferrari, Silvia and Cribari-Neto, Francisco},
  journal={Journal of applied statistics},
  volume={31},
  number={7},
  pages={799--815},
  year={2004},
  publisher={Taylor \& Francis}
}

@inproceedings{kull2017beta,
  title={Beta calibration: a well-founded and easily implemented improvement on logistic calibration for binary classifiers},
  author={Kull, Meelis and Silva Filho, Telmo and Flach, Peter},
  booktitle={Artificial intelligence and statistics},
  pages={623--631},
  year={2017},
  organization={PMLR}
}

@inproceedings{ma2024beta,
  title={Beta-LR: Interpretable logical reasoning based on beta distribution},
  author={Ma, Yizhuo and Qin, Ke and Liang, Shuang},
  booktitle={Findings of the Association for Computational Linguistics: NAACL 2024},
  pages={1945--1955},
  year={2024}
}

@inproceedings{reichl2010logarithmic,
  title={The logarithmic nature of QoE and the role of the Weber-Fechner law in QoE assessment},
  author={Reichl, Peter and Egger, Sebastian and Schatz, Raimund and D'Alconzo, Alessandro},
  booktitle={2010 IEEE international conference on communications},
  pages={1--5},
  year={2010},
  organization={IEEE}
}

@article{2025datawiseagent,
  title={DatawiseAgent: A Notebook-Centric LLM Agent Framework for Automated Data Science},
  author={You, Ziming and Zhang, Yumiao and Xu, Dexuan and Lou, Yiwei and Yan, Yandong and Wang, Wei and Zhang, Huaming and Huang, Yu},
  journal={arXiv preprint arXiv:2503.07044},
  year={2025}
}

@article{bi2025redundancy,
  title={Redundancy as a Structural Information Principle for Learning and Generalization},
  author={Bi, Yuda and Zhu, Ying and Calhoun, Vince D.},
  journal={arXiv preprint arXiv:2510.10938},
  year={2025},
  url={https://arxiv.org/abs/2510.10938}
}

@inproceedings{jing2025dsbench,
  title={DSBench: How Far Are Data Science Agents from Becoming Data Science Experts?},
  author={Jing, Liqiang and Huang, Zhehui and Wang, Xiaoyang and Yao, Wenlin and Yu, Wenhao and Ma, Kaixin and Zhang, Hongming and Du, Xinya and Yu, Dong},
  booktitle={International Conference on Learning Representations},
  year={2025}
}

@inproceedings{hu2024infiagent,
  title={InfiAgent-DABench: Evaluating Agents on Data Analysis Tasks},
  author={Hu, Xueyu and Zhao, Ziyu and Wei, Shuang and others},
  booktitle={International Conference on Machine Learning},
  year={2024}
}

@inproceedings{yang2024matplotagent,
  title={MatPlotAgent: Method and Evaluation for LLM-Based Agentic Scientific Data Visualization},
  author={Yang, Zhiyu and Zhou, Zihan and Wang, Shuo and others},
  booktitle={Findings of ACL},
  year={2024}
}

@article{ong2024routellm,
  title={Routellm: Learning to route llms with preference data},
  author={Ong, Isaac and Almahairi, Amjad and Wu, Vincent and Chiang, Wei-Lin and Wu, Tianhao and Gonzalez, Joseph E and Kadous, M Waleed and Stoica, Ion},
  journal={arXiv preprint arXiv:2406.18665},
  year={2024}
}

@article{tsiourvas2025causal,
  title={Causal {LLMs} routing: End-to-end regret minimization from observational data},
  author={Tsiourvas, Asterios and Sun, Wei and Perakis, Georgia},
  journal={arXiv preprint arXiv:2505.16037},
  year={2025}
}

@inproceedings{wang2025mixture,
  title={Mixture-of-Agents Enhances Large Language Model Capabilities},
  author={Wang, Junlin and Wang, Jue and Athiwaratkun, Ben and Zhang, Ce and Zou, James Y.},
  booktitle={International Conference on Learning Representations},
  year={2025}
}

@article{huang2025routereval,
  title={{RouterEval}: A comprehensive benchmark for routing {LLMs} to explore model-level scaling up in {LLMs}},
  author={Huang, Zhongzhan and Ling, Guoming and Lin, Yupei and Chen, Yandong and Zhong, Shanshan and Wu, Hefeng and Lin, Liang},
  journal={arXiv preprint arXiv:2503.10657},
  year={2025}
}

@article{hong2023metagpt,
  title={MetaGPT: Meta Programming for A Multi-Agent Collaborative Framework},
  author={Hong, Sirui and Zhuge, Mingchen and Chen, Jiaqi and Zheng, Xiawu and Cheng, Yuheng and Zhang, Ceyao and Wang, Jinlin and Wang, Zili and Yau, Steven Ka Shing and Lin, Zijuan and Zhou, Liyang and Ran, Chenyu and Xiao, Lingfeng and Wu, Chenglin and Schmidhuber, J{\"u}rgen},
  journal={arXiv preprint arXiv:2308.00352},
  year={2023},
  url={https://arxiv.org/abs/2308.00352}
}

@article{carson2003nist,
  author  = {Mark Carson and Darrin Santay},
  title   = {NIST Net: A Linux-Based Network Emulation Tool},
  journal = {Computer Communication Review},
  volume  = {33},
  number  = {3},
  pages   = {111--126},
  year    = {2003}
}

@manual{hemminger2005netem,
  author = {Stephen Hemminger},
  title  = {tc-netem(8) Linux Manual Page},
  year   = {2005},
  note   = {Linux Network Emulator (NetEm)}
}

@article{gouveia2020kollaps,
  author  = {Paulo Gouveia and Jo{\~a}o Neves and Carlos Segarra and Luca Liechti and
             Shady Issa and Valerio Schiavoni and Miguel Matos},
  title   = {Kollaps: Decentralized and Dynamic Topology Emulation},
  journal = {arXiv preprint arXiv:2004.02253},
  year    = {2020}
}

@misc{azure_foundry_deployment_types,
  title        = {Understanding Deployment Types in Microsoft Foundry Models},
  author       = {{Microsoft}},
  year         = {2026},
  howpublished = {\url{https://learn.microsoft.com/en-us/azure/foundry/foundry-models/concepts/deployment-types}},
  note         = {Accessed: 2026-03-13}
}

@misc{vertex_deployments_endpoints,
  title        = {Deployments and Endpoints},
  author       = {{Google Cloud}},
  year         = {2026},
  howpublished = {\url{https://docs.cloud.google.com/vertex-ai/generative-ai/docs/learn/locations}},
  note         = {Accessed: 2026-03-13}
}

@misc{vertex_standard_paygo,
  title        = {Standard PayGo},
  author       = {{Google Cloud}},
  year         = {2026},
  howpublished = {\url{https://docs.cloud.google.com/vertex-ai/generative-ai/docs/standard-paygo}},
  note         = {Accessed: 2026-03-13}
}

@misc{bedrock_cross_region,
  title        = {Increase Throughput with Cross-Region Inference},
  author       = {{Amazon Web Services}},
  year         = {2026},
  howpublished = {\url{https://docs.aws.amazon.com/bedrock/latest/userguide/cross-region-inference.html}},
  note         = {Accessed: 2026-03-13}
}

@misc{bedrock_geo_cross_region,
  title        = {Geographic Cross-Region Inference},
  author       = {{Amazon Web Services}},
  year         = {2026},
  howpublished = {\url{https://docs.aws.amazon.com/bedrock/latest/userguide/geographic-cross-region-inference.html}},
  note         = {Accessed: 2026-03-13}
}

@misc{bedrock_regions_compatibility,
  title        = {Regional Availability - Amazon Bedrock},
  author       = {{Amazon Web Services}},
  year         = {2026},
  howpublished = {\url{https://docs.aws.amazon.com/bedrock/latest/userguide/models-region-compatibility.html}},
  note         = {Accessed: 2026-03-13}
}

@misc{openai_data_controls,
  title        = {Data Controls in the OpenAI Platform},
  author       = {{OpenAI}},
  year         = {2026},
  howpublished = {\url{https://developers.openai.com/api/docs/guides/your-data/}},
  note         = {Accessed: 2026-03-13}
}

@misc{cloudflare_workers_ai,
  title        = {Cloudflare Workers AI Overview},
  author       = {{Cloudflare}},
  year         = {2026},
  howpublished = {\url{https://developers.cloudflare.com/workers-ai/}},
  note         = {Accessed: 2026-03-13}
}

@misc{cloudflare_network,
  title        = {Cloudflare Global Network | Data Center Locations},
  author       = {{Cloudflare}},
  year         = {2026},
  howpublished = {\url{https://www.cloudflare.com/network/}},
  note         = {Accessed: 2026-03-13}
}

@misc{fireworks_regions,
  title        = {Regions - Fireworks AI Docs},
  author       = {{Fireworks AI}},
  year         = {2026},
  howpublished = {\url{https://docs.fireworks.ai/deployments/regions}},
  note         = {Accessed: 2026-03-13}
}

@misc{openrouter_quickstart,
  title        = {OpenRouter Quickstart Guide},
  author       = {{OpenRouter}},
  year         = {2026},
  howpublished = {\url{https://openrouter.ai/docs/quickstart}},
  note         = {Accessed: 2026-03-13}
}

@misc{portkey_fallbacks,
  title        = {Fallbacks - Portkey Docs},
  author       = {{Portkey}},
  year         = {2026},
  howpublished = {\url{https://portkey.ai/docs/product/ai-gateway/fallbacks}},
  note         = {Accessed: 2026-03-13}
}

@misc{litellm_auto_routing,
  title        = {Auto Routing - LiteLLM},
  author       = {{LiteLLM}},
  year         = {2026},
  howpublished = {\url{https://docs.litellm.ai/docs/proxy/auto_routing}},
  note         = {Accessed: 2026-03-13}
}

@misc{openrouter_provider_routing,
  title        = {Intelligent Multi-Provider Request Routing},
  author       = {{OpenRouter}},
  year         = {2026},
  howpublished = {\url{https://openrouter.ai/docs/guides/routing/provider-selection}},
  note         = {Accessed: 2026-03-13}
}

@article{ding2024hybrid,
  title={Hybrid {LLM}: Cost-efficient and quality-aware query routing},
  author={Ding, Dujian and Mallick, Ankur and Wang, Chi and Sim, Robert and Mukherjee, Subhabrata and Ruhle, Victor and Lakshmanan, Laks VS and Awadallah, Ahmed Hassan},
  journal={arXiv preprint arXiv:2404.14618},
  year={2024}
}

@article{yue2023large,
  title={Large language model cascades with mixture of thoughts representations for cost-efficient reasoning},
  author={Yue, Murong and Zhao, Jie and Zhang, Min and Du, Liang and Yao, Ziyu},
  journal={arXiv preprint arXiv:2310.03094},
  year={2023}
}

@inproceedings{kwon2023vllm,
  title={Efficient Memory Management for Large Language Model Serving with {PagedAttention}},
  author={Kwon, Woosuk and Li, Zhuohan and Zhuang, Siyuan and Sheng, Ying and Zheng, Lianmin and Yu, Cody Hao and Gonzalez, Joseph and Zhang, Hao and Stoica, Ion},
  booktitle={Proceedings of the 29th symposium on operating systems principles},
  pages={611--626},
  year={2023}
}

@inproceedings{jiang2023llmblender,
  title={{LLM}-Blender: Ensembling Large Language Models with Pairwise Ranking and Generative Fusion},
  author={Jiang, Dongfu and Ren, Xiang and Lin, Bill Yuchen},
  booktitle={Proceedings of the 61st Annual Meeting of the Association for Computational Linguistics (Volume 1: Long Papers)},
  pages={14165--14178},
  year={2023}
}

@article{wang2023selfconsistency,
  title={Self-Consistency Improves Chain of Thought Reasoning in Language Models},
  author={Wang, Xuezhi and Wei, Jason and Schuurmans, Dale and Le, Quoc and Chi, Ed and Narang, Sharan and Chowdhery, Aakanksha and Zhou, Denny},
  journal={arXiv preprint arXiv:2203.11171},
  year={2022}
}

@article{gupta2024language,
  title={Language model cascades: Token-level uncertainty and beyond},
  author={Gupta, Neha and Narasimhan, Harikrishna and Jitkrittum, Wittawat and Rawat, Ankit Singh and Menon, Aditya Krishna and Kumar, Sanjiv},
  journal={arXiv preprint arXiv:2404.10136},
  year={2024}
}

@article{zhang2024treacle,
  title={Efficient Contextual {LLM} Cascades through Budget-Constrained Policy Learning},
  author={Zhang, Xuechen and Huang, Zijian and Taga, Ege Onur and Joe-Wong, Carlee and Oymak, Samet and Chen, Jiasi},
  journal={Advances in Neural Information Processing Systems},
  volume={37},
  pages={91691--91722},
  year={2024}
}

@inproceedings{zhong2024distserve,
  title={$\{$DistServe$\}$: Disaggregating prefill and decoding for goodput-optimized large language model serving},
  author={Zhong, Yinmin and Liu, Shengyu and Chen, Junda and Hu, Jianbo and Zhu, Yibo and Liu, Xuanzhe and Jin, Xin and Zhang, Hao},
  booktitle={18th USENIX Symposium on Operating Systems Design and Implementation (OSDI 24)},
  pages={193--210},
  year={2024}
}

@article{sheng2024slora,
  title={{SLoRA}: Scalable Serving of Thousands of {LoRA} Adapters},
  author={Sheng, Ying and Cao, Shiyi and Li, Dacheng and Hooper, Coleman and Lee, Nicholas and Yang, Shuo and Chou, Christopher and Zhu, Banghua and Zheng, Lianmin and Keutzer, Kurt and others},
  journal={Proceedings of Machine Learning and Systems},
  volume={6},
  pages={296--311},
  year={2024}
}

@inproceedings{wu2024autogen,
  title={Autogen: Enabling next-gen LLM applications via multi-agent conversations},
  author={Wu, Qingyun and Bansal, Gagan and Zhang, Jieyu and Wu, Yiran and Li, Beibin and Zhu, Erkang and Jiang, Li and Zhang, Xiaoyun and Zhang, Shaokun and Liu, Jiale and others},
  booktitle={First conference on language modeling},
  year={2024}
}

\appendix
\section{Cross-Region Inference in Practice}
\label{sec:cross-region-gap}
This section documents publicly available evidence that cross-region inference is already a
practical deployment pattern in modern LLM serving systems. Modern platforms expose several
distinct abstractions that decouple API ingress from the physical location of inference execution.
Although their implementations differ, these systems consistently allow the execution region to vary
with deployment configuration, routing policy, capacity conditions, or compliance constraints.

From a systems perspective, existing platforms can be grouped into three representative categories.

\paragraph{(1) Cloud-region routing platforms}

Some LLM serving systems expose global endpoints that dynamically route requests across multiple
cloud regions. In these systems, the developer interacts with a stable API endpoint while the
provider determines the actual inference location according to internal routing policies such as
capacity availability or load balancing.

Azure Foundry provides three deployment types: \emph{Global}, \emph{Data Zone}, and
\emph{Regional}. Under Global deployment, requests may be processed in any Azure region where the
model is available. Data Zone deployments restrict processing within a specified geography such as
the US or EU, while Regional deployments keep execution within a specific region
\citep{azure_foundry_deployment_types}.

Vertex AI exposes a similar abstraction through \emph{regional endpoints} and a \emph{global
endpoint}. The global endpoint draws from a larger multi-region capacity pool and therefore offers
higher availability and burst tolerance than a single region. However, the platform does not expose
the actual processing region of a request to the developer
\citep{vertex_deployments_endpoints, vertex_standard_paygo}.

Amazon Bedrock provides cross-region execution through \emph{inference profiles}. These profiles
support both \emph{geographic cross-region inference}, which restricts execution within a specific
geography, and \emph{global cross-region inference}, which allows requests to be routed across
commercial AWS regions worldwide
\citep{bedrock_cross_region, bedrock_geo_cross_region, bedrock_regions_compatibility}. Bedrock also
documents that the final inference region can be observed through service telemetry such as
CloudTrail.

In all of these systems, region selection is determined by the infrastructure layer.

\paragraph{(2) Region pinning and data residency}

A second design treats region as an explicit configuration parameter. In this setting, developers choose the processing region during deployment, and inference is performed within that region.

OpenAI provides regional processing and data residency features at the project level. Developers
select a project region and, where supported, access the corresponding regional endpoint
\citep{openai_data_controls}. While this model does not expose dynamic cross-region routing in the
same way as Azure, Vertex AI, or Bedrock, it confirms that physical processing location is treated
as an operational concern in production LLM APIs.

\paragraph{(3) Distributed inference infrastructure}

A third design emerges in systems that provide inference over geographically distributed GPU
infrastructure.

Cloudflare Workers AI runs inference on GPUs deployed across Cloudflare's global edge network,
allowing models to execute close to user locations
\citep{cloudflare_workers_ai, cloudflare_network}. Fireworks provides multi-region deployments with
scopes such as \texttt{GLOBAL}, \texttt{US}, \texttt{EUROPE}, and \texttt{APAC}. Deployments may run
in any region within the selected multi-region scope
\citep{fireworks_regions}. In both cases, physical location is treated as part of the serving
substrate, not as a fixed endpoint.

\begin{table*}[t]
\caption{Publicly documented cross-region inference mechanisms in representative LLM serving
platforms. Existing systems already treat execution location as an operational variable at the
infrastructure layer.}
\centering
\footnotesize
\setlength{\tabcolsep}{4pt}
\renewcommand{\arraystretch}{1.08}
\begin{tabularx}{\textwidth}{
>{\raggedright\arraybackslash}p{1.7cm}
>{\raggedright\arraybackslash}X
>{\raggedright\arraybackslash}X
>{\raggedright\arraybackslash}X}
\toprule
\textbf{Platform} &
\textbf{Public cross-region support} &
\textbf{Developer control / visibility} &
\textbf{Primary platform objective} \\
\midrule

Azure Foundry
&
Global, Data Zone, and Regional deployments. Global and Data Zone deployments may route requests
across eligible datacenters.
&
Developers choose the deployment type. Global deployments may process prompts and responses in any
eligible Azure region, while Regional deployments keep processing in the deployment region.
&
Increase quota, throughput, model availability, and geography-level compliance. \\
\hline

Vertex AI
&
Regional endpoints and a global endpoint. The global endpoint routes traffic over a larger
multi-region capacity pool.
&
Developers choose regional versus global endpoints, but the actual processing region of the global
endpoint is not exposed to the user.
&
Improve availability, reduce 429 errors, and absorb bursty demand through capacity-aware routing. \\
\hline

Amazon Bedrock
&
Geographic and Global cross-region inference profiles route requests across multiple AWS Regions.
&
Developers choose an inference profile. The actual inference region can be observed through
CloudTrail.
&
Increase throughput, absorb traffic bursts, and satisfy geography-level compliance constraints. \\
\hline

OpenAI API
&
Regional processing and data residency at the project level.
&
Developers explicitly choose a project region and, where supported, use the corresponding regional
endpoint.
&
Support compliance, regional processing, and in some cases lower latency. \\
\hline

Cloudflare Workers AI
&
Inference runs on Cloudflare's global network with distributed GPU deployment.
&
Developers use a global network abstraction, while the concrete inference region is left unspecified.
&
Serve requests closer to users and improve responsiveness and resilience. \\
\hline

Fireworks AI
&
Default multi-region deployments over \texttt{GLOBAL}, \texttt{US}, \texttt{EUROPE}, and
\texttt{APAC}. Deployments may run in any region within the selected scope.
&
Developers can use a multi-region deployment or pin the deployment to a single region.
&
Elastic scaling, availability, regional compliance, and deployment near users. \\
\bottomrule
\end{tabularx}
\label{tab:cross_region_platforms}
\end{table*}

Across these designs, the physical execution region of an LLM request is therefore not fixed
\emph{a priori}. Instead, it may vary with deployment configuration, routing policy, capacity
conditions, or compliance constraints. This observation directly supports our formulation in which
the execution environment of a role invocation represents a runtime context with potentially
different latency, execution reliability, and downstream connectivity characteristics.

Table~\ref{tab:cross_region_platforms} shows that cross-region inference is already supported by
major model serving platforms. As we can see, current cloud inference platforms already consider region, but they do so mainly for infrastructure-level goals such as throughput, burst absorption, availability, and compliance \citep{azure_foundry_deployment_types, vertex_deployments_endpoints, vertex_standard_paygo, bedrock_cross_region, bedrock_geo_cross_region, openai_data_controls, cloudflare_workers_ai, fireworks_regions}. By contrast, third-party gateways and model access layers usually expose LLM choice or provider choice as the main control surface, emphasizing fallback, load balancing, and provider-level routing, with less direct attention to the execution region itself \citep{openrouter_quickstart, openrouter_provider_routing, portkey_fallbacks, litellm_auto_routing}. What remains weakly addressed is \emph{request-aware, region-aware execution} for agent workloads. Task-level performance depends not only on LLM choice, but also on where inference runs. HACO introduces an application-level coordination layer for agent execution. It operates above existing provider routing mechanisms and makes task-level decisions about whether to launch redundant executions, which candidates to include, and how to balance latency, system reliability, and output quality. This design incorporates query characteristics and execution region conditions into the coordination process, enabling agent systems to utilize regional diversity during inference.

\section{Comparison with Related Runtime and Inference-Time Paradigms}
\label{appcom}
\begin{table}[h]
\centering
\footnotesize
\caption{HACO compared with related runtime and inference-time paradigms.}
\label{tab:taxonomy}
\begin{tabularx}{\linewidth}{p{2.0cm} p{2.4cm} p{2.4cm} X}
\toprule
\textbf{Paradigm} & \textbf{Control point} & \textbf{Execution pattern} & \textbf{Relation to HACO} \\
\midrule
LLM routing~\citep{ding2024hybrid,ong2024routellm} &
Single-executor selection &
One model or endpoint per request &
Routing chooses one executor for a request, usually to balance quality and cost. HACO operates at the role-invocation allocation point and selects an adaptive hedge set when a single execution is not sufficient for the target reliability level. \\
\midrule
Fallback or cascade~\citep{yue2023large,gupta2024language,zhang2024treacle} &
Post-hoc escalation or recovery &
Sequential backup calls &
Fallback and cascade methods invoke a stronger or backup executor after failure, low confidence, or insufficient consistency. HACO allocates redundancy before execution using estimated capability, failure, latency, and link conditions. \\
\midrule
Batch inference~\citep{kwon2023vllm,zhong2024distserve,sheng2024slora} &
Serving-batch construction &
Many requests processed together &
Batch inference improves serving throughput or hardware utilization across requests. HACO controls reliability for each role invocation before the launched calls are served. \\
\midrule
Multi-output aggregation~\citep{jiang2023llmblender,wang2025mixture,wang2023selfconsistency} &
Output selection or synthesis &
Multiple samples, agents, or models &
Best-of-$N$, self-consistency, MoA, and ensemble methods mainly improve final answer quality through selection, voting, ranking, or synthesis. HACO decides which executions to launch under runtime uncertainty; aggregation can be applied after the hedge set produces outputs as a potential future research direction. \\
\midrule
\textbf{HACO} &
\textbf{Role-invocation allocation} &
\textbf{Adaptive parallel execution} &
\textbf{HACO selects a reliability-constrained hedge set over candidate agent instances and stops allocation when the target reliability level is reached.} \\
\bottomrule
\end{tabularx}
\end{table}

Table~\ref{tab:taxonomy} clarifies the control point addressed by HACO. Existing paradigms typically operate at single-route selection, post-failure recovery, serving-batch construction, or output aggregation. HACO instead addresses the preceding allocation problem for each role invocation. Specifically, before execution, it decides which candidate agent instances should be activated and when the selected hedge set has reached the required reliability level. This formulation makes output aggregation a possible downstream policy.

\paragraph{Role of Redundancy in HACO.}
Figure~\ref{fig:novelty} further clarifies the role of redundancy in HACO.
In contrast to methods that use parallelism mainly for post-hoc output selection, synthesis, or serving efficiency, HACO uses redundancy as a runtime control mechanism at the role-invocation level.
For the current invocation, hedged execution improves reliability by tolerating candidate failures, high latency, and unstable service or network conditions.
For future invocations, the executed hedge set provides additional traces for experience harvesting, including quality, success or failure, latency, and network statistics.
These observations update candidate and link profiles, allowing HACO to improve later routing and hedge-set allocation.
Thus, redundancy in HACO serves both as an immediate protection mechanism and as a data-collection mechanism for adaptive future decisions.

\begin{figure}[h]
    \centering
    \includegraphics[width=1.0\textwidth]{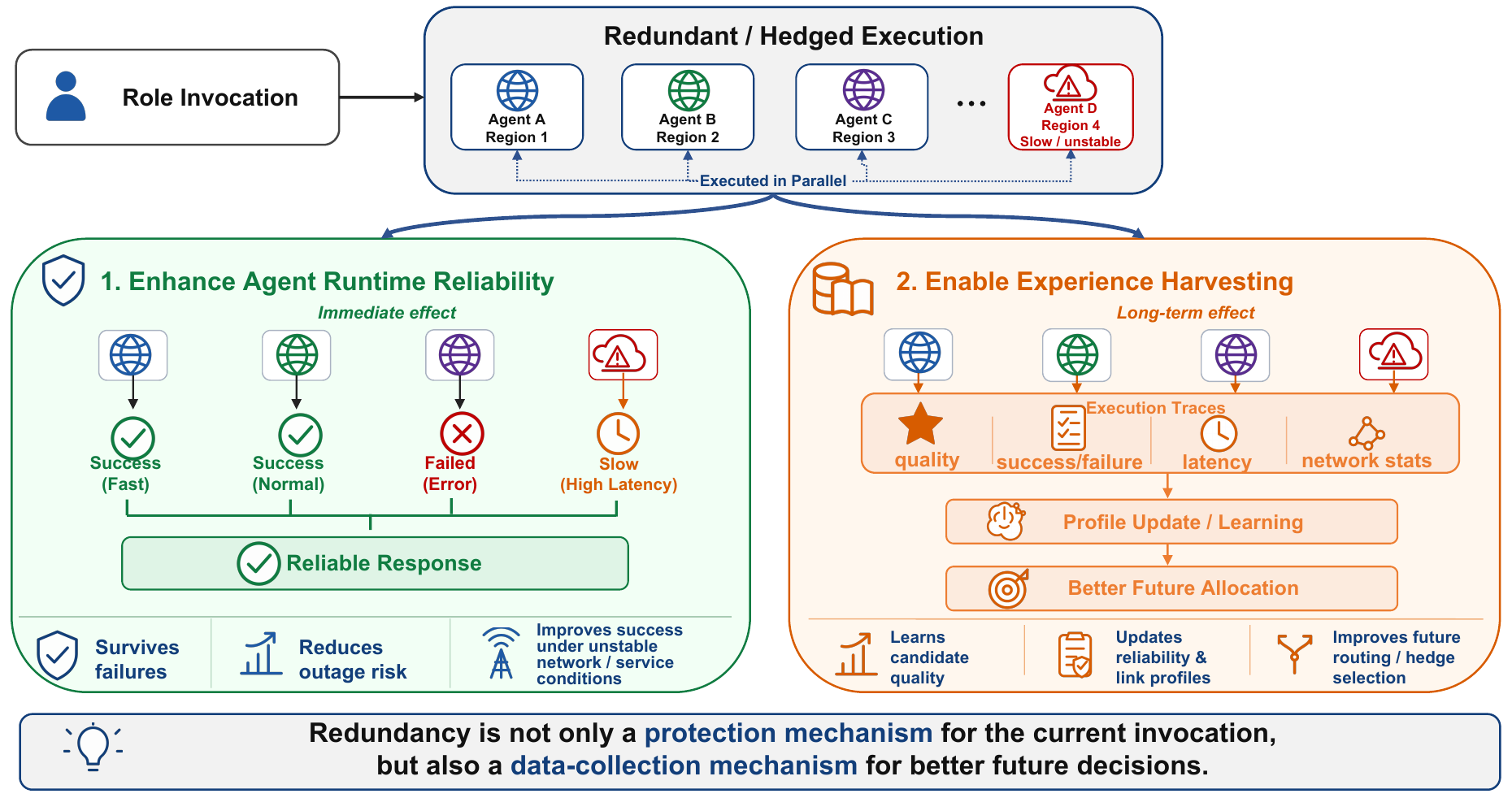}
    \caption{
    \textbf{Role of redundancy in HACO.}
    Redundant or hedged execution improves reliability for the current role invocation while also collecting execution traces for future allocation.
    HACO records quality, success/failure, latency, and network statistics from executed candidates, updates candidate and link profiles, and uses these profiles to improve later routing and hedge-set selection.
    }
    \label{fig:novelty}
\end{figure}
\section{Experimental Setup Details}

\subsection{MAS-DATAWISE settings}
\label{sec: datawise_mas}

Our multi-agent system is implemented as an extension of \textsc{DATAWISE}. 
In particular, we build on the original \textsc{DATAWISE}~\cite{2025datawiseagent} framework and adapt it from a single-system workflow into a collaborative multi-agent setting, while preserving its controller logic and execution structure. 
We use this system as the primary testbed in our experiments because it provides a representative controller with explicit state transitions, iterative code execution, and built-in self-debugging behaviors. 
At the same time, HACO is not tied to this specific controller design. 
Its routing and redundancy allocation mechanisms are agnostic to the internal orchestration policy of the underlying MAS, and can therefore be applied not only to complex finite-state-machine (FSM) controllers such as \textsc{DATAWISE}, but also to simpler workflow-style controllers such as \textsc{MetaGPT}~\cite{hong2023metagpt}, as well as other controller paradigms.

\begin{figure}[ht]
    \centering
    \includegraphics[width=1.0\textwidth]{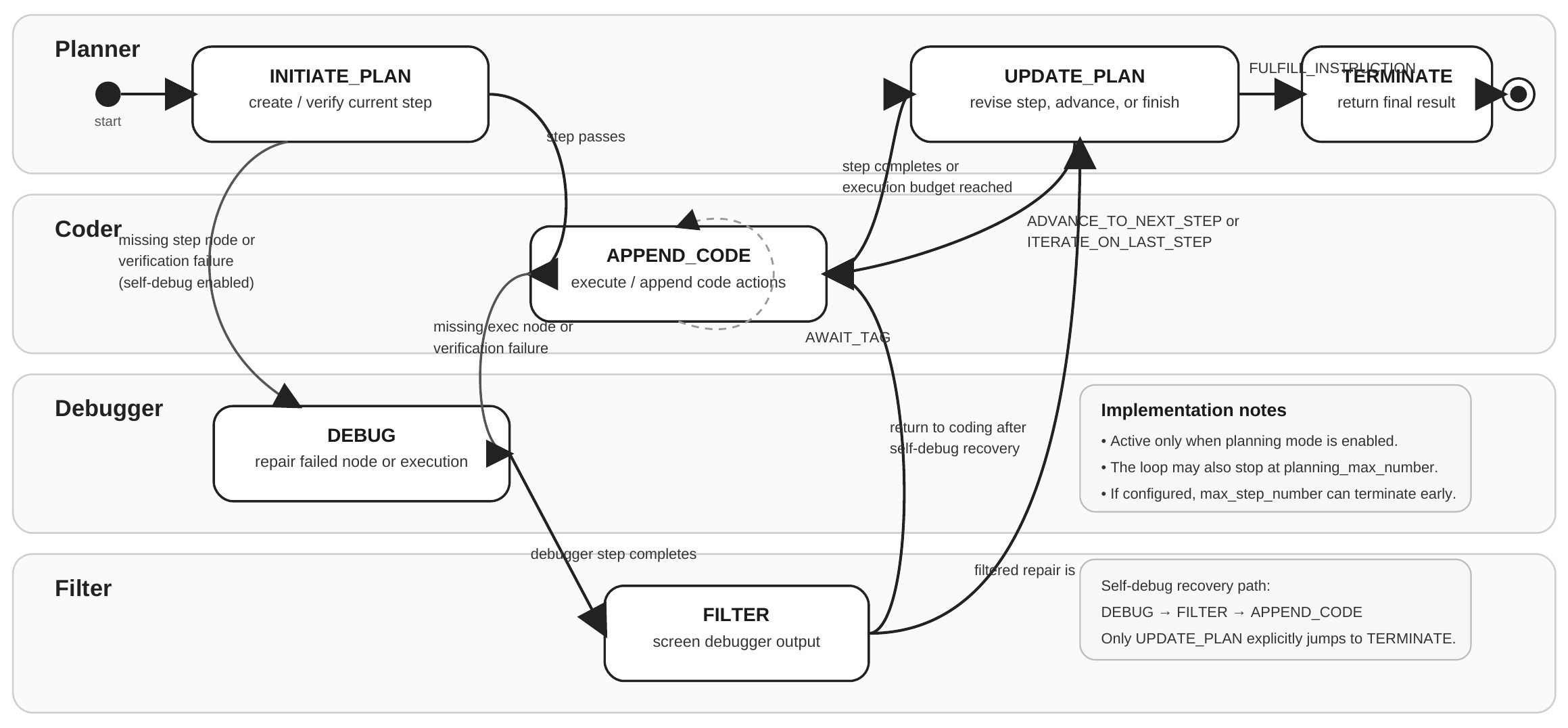}
    \caption{\textbf{Role-oriented view of the DATAWISE concurrent FSM workflow.} The underlying controller is a six-state finite-state machine, while the execution dynamics can be organized into four role types: planner, coder, debugger, and filter. The figure highlights role collaboration and the self-debug recovery path.}
    \label{fig:datawise_role_workflow}
\end{figure}

Fig.~\ref{fig:datawise_role_workflow} illustrates the DATAWISE concurrent FSM workflow from a role-oriented perspective. The underlying controller is a six-state finite-state machine with states INITIATE\_PLAN, APPEND\_CODE, DEBUG, FILTER, UPDATE\_PLAN, and TERMINATE. Functionally, however, the workflow is organized around four role types. The planner initializes and updates task steps, the coder performs iterative code generation and execution, the debugger repairs failures arising during execution or verification, and the filter cleans debugging outputs before execution resumes. The coding stage is controlled by explicit action signals: AWAIT\_TAG prolongs the current execution stage, whereas completion signals return control to the planner. At the workflow level, the planner determines whether to iterate on the current step, advance to the next step, or terminate the process through the signals ITERATE\_ON\_LAST\_STEP, ADVANCE\_TO\_NEXT\_STEP, and FULFILL\_INSTRUCTION. Together, these mechanisms form a structured self-debugging workflow in which recovery follows the path DEBUG $\rightarrow$ FILTER $\rightarrow$ APPEND\_CODE.

\subsection{Simulating Environment}
Fig.~\ref{fig:setting1} provides an overview of the simulated experimental setup. It combines the
zone-aware network topology over Local, Regional, and Global execution zones with the
role-specific candidate pool used by the MAS. The network side specifies latency, jitter,
bandwidth, and loss for each logical inter-zone link, while the candidate side shows the
LLM backbone, replica count, failure rates, and temperature settings for each role. The following
two sections detail these two components: the topology-based network simulator and the
candidate-pool configuration.
\begin{figure}[h]
    \centering
    \includegraphics[width=1.0\textwidth]{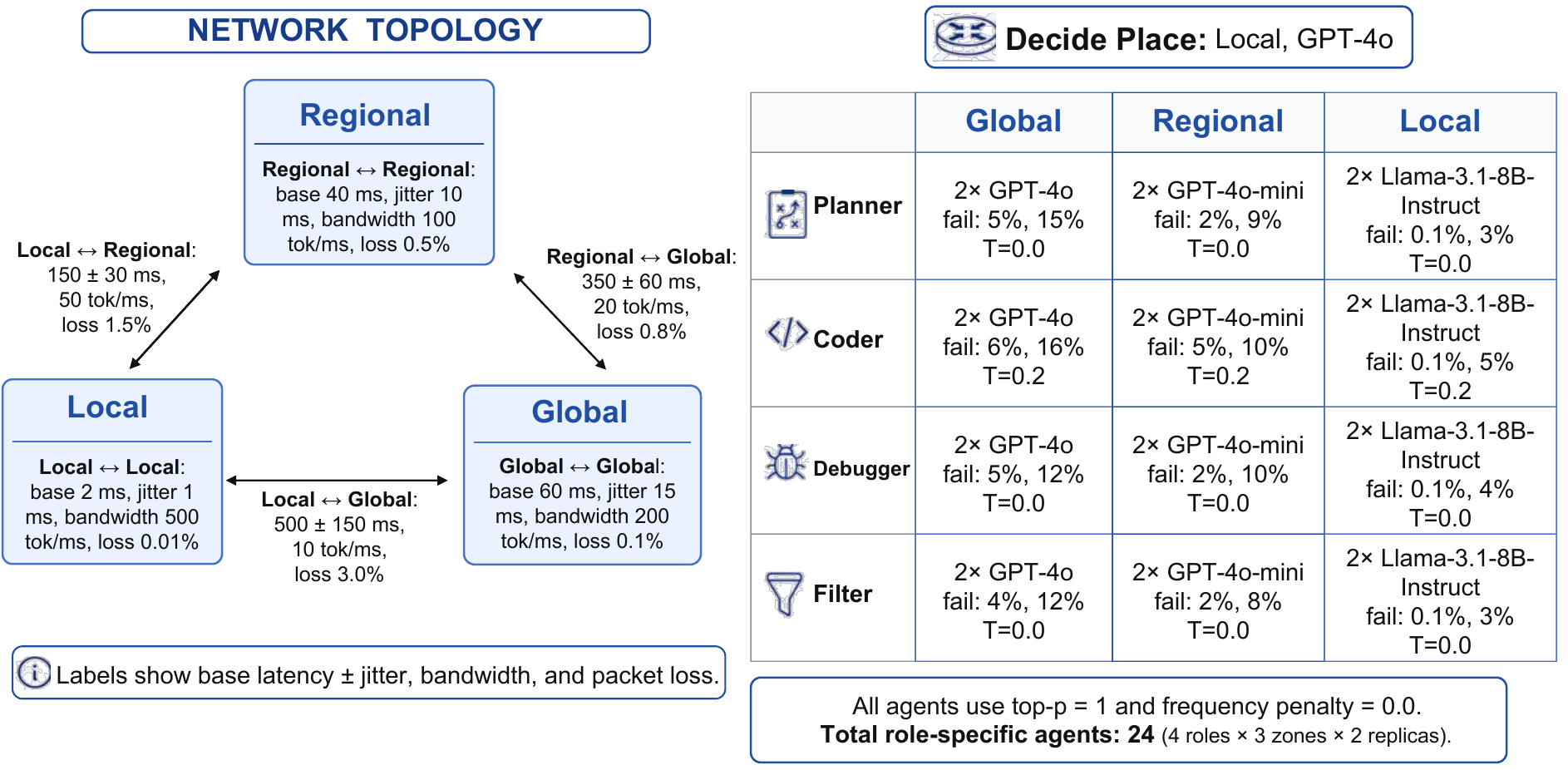}
    \caption{\textbf{Simulating environment overview.}}
    \label{fig:setting1}
\end{figure}

\subsubsection{Network Simulation Details}
\label{app:network_simulation}

We simulate inter-candidate communication using a \textbf{message-level network model}.
Specifically, we directly approximate the end-to-end
latency and failure behavior of transmitting a complete message (e.g., a list of
notebook cells) between candidates. 
This design follows the parameterized network-emulation
paradigm used in prior systems such as NIST Net and Linux NetEm, which models the network
behavior using delay, jitter, bandwidth constraints, and loss parameters
\cite{carson2003nist,hemminger2005netem}. 
It is also consistent with the view that, for distributed applications, the most important network properties are emergent end-to-end characteristics such as latency, bandwidth, packet loss, and jitter, while the full internal state of routers and switches is less directly relevant \cite{gouveia2020kollaps}.

The topology in Fig.~\ref{fig:setting1} and Table~\ref{tab:candidate_pool_default_network} is reported in simulator configuration units: latency in milliseconds and bandwidth in tokens/ms. The simulator computes message-level communication delay using these units. Before the resulting observations are used by HACO's effective-latency model in Eq.~\eqref{eq:effective-latency}, payload size is represented in bytes, bandwidth statistics are converted to bytes/s, and latency values are converted to seconds.

\paragraph{Network topology.}
We define a set of execution zones
\[
\mathcal{Z} = \{\texttt{Global}, \texttt{Regional}, \texttt{Local}\}.
\]
Each execution-zone pair is associated with a logical link whose properties are specified in
a topology matrix. Each link is parameterized by:
(i) base latency $\mu$ in milliseconds,
(ii) jitter standard deviation $\sigma$ in milliseconds,
(iii) bandwidth $B$ in tokens/ms, and
(iv) loss rate $p \in [0,1]$.

\paragraph{Message size.}
Given a message at step $i$ consisting of a list of notebook cells, we estimate its payload size
$M^{(i)}$ (in bytes) by serializing the cells into JSON.

\paragraph{Transmission delay.}
For a successful transmission at step $i$, the serialization/transmission delay is
\begin{equation}
T_{\mathrm{tx}}^{(i)} = \frac{M^{(i)}}{4B},
\end{equation}
where $B$ is the configured bandwidth in tokens/ms and we assume $1$ token $\approx 4$ bytes.

\paragraph{Path delay with jitter.}
The path delay is modeled as a base latency perturbed by zero-mean Gaussian noise:
\begin{equation}
T_{\mathrm{path}}^{(i)} = \max\!\left(T_{\min}(\mu),\ \mu + \xi^{(i)}\right),
\qquad
\xi^{(i)} \sim \mathcal{N}(0,\sigma^2),
\end{equation}
where $T_{\min}(\mu)$ is a latency floor that prevents unrealistic near-zero delays
under negative jitter realizations. In particular, we use a \emph{base-relative}
lower bound:
\begin{equation}
T_{\min}(\mu) = \max(\epsilon_{\rm floor},\ \zeta \mu),
\end{equation}
where $\epsilon_{\rm floor}$ is a small absolute lower bound and $\zeta \in (0,1)$ is a fixed
ratio. In our implementation, we use $\epsilon_{\rm floor}=1$ ms and $\zeta = 0.2$.
This design preserves the scale differences between short-range and long-range links:
for example, local links may have lower absolute floors than cross-region links,
while all links remain protected against invalid near-zero latency samples.

\paragraph{Communication failure and timeout.}
We model failures at the \emph{message level}. For each transmission at step $i$, a failure occurs
independently with probability $p$. We abstract away lower-layer retransmission and transport
protocol dynamics, and approximate the \emph{end effect} of unsuccessful communication at the
application level: a failed transmission incurs a fixed timeout penalty,
\begin{equation}
t_{\mathrm{net}}^{(i)} = T_{\mathrm{timeout}},
\end{equation}
where $T_{\mathrm{timeout}} = 10$ seconds in our implementation. This corresponds to application-level
request timeout behavior, such as an RPC timeout, and does not model packet-level network timeout.

\paragraph{Total communication latency.}
Otherwise (i.e., if the transmission succeeds), the communication latency is computed as
\begin{equation}
t_{\mathrm{net}}^{(i)} = T_{\mathrm{tx}}^{(i)} + T_{\mathrm{path}}^{(i)}.
\end{equation}

Equivalently, the complete model can be written as
\begin{equation}
t_{\mathrm{net}}^{(i)} =
\begin{cases}
T_{\mathrm{timeout}}, & \text{with probability } p,\\[4pt]
T_{\mathrm{tx}}^{(i)} + T_{\mathrm{path}}^{(i)}, & \text{with probability } 1-p,
\end{cases}
\qquad
\xi^{(i)} \sim \mathcal{N}(0,\sigma^2).
\end{equation}

\paragraph{Discussion.}
Our simulator is intentionally lightweight. It does not model packet-level mechanisms
such as retransmissions, congestion control, reordering, or queueing disciplines.
Instead, these effects are abstracted through stochastic delay and application-level
timeout penalties. The use of a base-relative latency floor further ensures that
jitter perturbs links realistically without collapsing long-range links into implausibly
small delays. This abstraction is sufficient for our setting, which focuses on
how heterogeneous inter-zone communication affects multi-agent coordination, while
leaving exact transport-protocol behavior outside the scope of the simulator.
\subsubsection{Candidate Pool Details}
\label{app:candidate_pool_details}

This section details the candidate-pool configurations used in the experiments.
The implementation loads role-specific candidates from the experiment
configuration and uses the active network-topology block for the topology-based
network simulator. The descriptions below focus on the resulting runtime
settings and omit configuration-file layout details.

\paragraph{Configurations for reliability under runtime uncertainty.}
The reliability-under-runtime-uncertainty experiments in
Fig.~\ref{fig:robustness} evaluate the degradation settings discussed in
Sec.~\ref{sec:experiments}: network degradation, candidate degradation, and
evaluator degradation. Each simulated setting changes one runtime dimension
after task index 50 while preserving the same role-invocation interface and MAS
workflow. Table~\ref{tab:robustness_configs} summarizes how each stress setting
differs from the overall benchmark performance configuration.

\begin{table*}[ht]
\caption{Configuration changes used for reliability evaluation under runtime
uncertainty. Each stress setting changes one axis of the runtime system:
candidate availability, network connectivity, or evaluator capability.}
\centering
\footnotesize
\setlength{\tabcolsep}{4pt}
\renewcommand{\arraystretch}{1.12}
\begin{tabularx}{\textwidth}{
>{\raggedright\arraybackslash}p{2.4cm}
>{\raggedright\arraybackslash}p{3.0cm}
>{\raggedright\arraybackslash}X}
\toprule
\textbf{Stress setting} & \textbf{Runtime setup} &
\textbf{Change relative to the overall benchmark performance configuration} \\
\midrule
Candidate degradation &
Global candidate outage &
Keeps the default network topology, but forces selected high-capability Global
candidates to fail. In particular, both Global planner replicas and one Global
coder replica are assigned failure rate 1.0. All other candidate definitions and
failure rates remain the same as in the overall benchmark performance pool. \\
\midrule
Network degradation &
Global-link outage &
Keeps the same 24 candidate definitions and candidate failure rates, but changes
the active network topology. All links involving the Global zone are made
unavailable with loss rate 1.0: Local--Global, Regional--Global, and
Global--Global. Local--Regional and Regional--Regional remain reachable with
60 ms / 10 ms / 150 tokens per ms / 0.003 loss and
15 ms / 4 ms / 220 tokens per ms / 0.0008 loss, respectively. \\
\midrule
Evaluator degradation &
Weak router/evaluator &
Keeps the overall benchmark performance candidate pool and default network
topology, but replaces the router/evaluator model with a Local
\texttt{meta-llama/llama-3.1-8b-instruct} instance. This setting stresses
methods whose final selection depends on evaluator quality, especially full
parallel generation followed by post-hoc scoring. \\
\bottomrule
\end{tabularx}
\label{tab:robustness_configs}
\end{table*}

\paragraph{Configuration for overall benchmark performance.}
The overall benchmark performance results in Fig.~\ref{fig:radar} compare HACO
with single-route, fixed-zone, and full-redundancy baselines under the
heterogeneous execution environment described in the main text. This setting
uses the full heterogeneous candidate pool. The pool contains four role types:
planner, coder, debugger, and filter. Each role has six candidate agent
instances: two Global candidates using \texttt{gpt-4o}, two Regional candidates
using \texttt{gpt-4o-mini}, and two Local candidates using
\texttt{meta-llama/llama-3.1-8b-instruct}. Thus, each execution zone contributes
two candidates per role and eight candidates in total across the four roles,
yielding 24 role candidates in the full pool. The router/evaluator agent is
placed in the Local zone and uses \texttt{gpt-4o}; it supports routing and
Best-of-\(N\) output scoring, but is not counted as a role candidate in the
pool size. All candidates use \texttt{top\_p}=1 and frequency penalty 0. The
planner, debugger, and filter roles use temperature 0.0, while the coder role
uses temperature 0.2. Table~\ref{tab:candidate_pool_main} summarizes the
simulated candidate pool used for the overall benchmark performance setting.

\begin{table*}[ht]
\caption{Simulated candidate pool for overall benchmark performance. Each failure rate corresponds to one replica in the same role-zone-model group.}
\centering
\footnotesize
\setlength{\tabcolsep}{4pt}
\renewcommand{\arraystretch}{1.08}
\begin{tabularx}{\textwidth}{
>{\raggedright\arraybackslash}p{1.8cm}
>{\raggedright\arraybackslash}X
>{\raggedright\arraybackslash}X
>{\raggedright\arraybackslash}X
>{\centering\arraybackslash}p{1.4cm}}
\toprule
\textbf{Role} &
\makecell[l]{\textbf{Global candidates}\\\texttt{gpt-4o}} &
\makecell[l]{\textbf{Regional candidates}\\\texttt{gpt-4o-mini}} &
\makecell[l]{\textbf{Local candidates}\\\texttt{Llama-3.1-8B}} &
\textbf{Temp.} \\
\midrule
Planner &
Failure rates 0.05, 0.15 &
Failure rates 0.02, 0.09 &
Failure rates 0.001, 0.03 &
0.0 \\
Coder &
Failure rates 0.06, 0.16 &
Failure rates 0.05, 0.10 &
Failure rates 0.001, 0.05 &
0.2 \\
Debugger &
Failure rates 0.05, 0.12 &
Failure rates 0.02, 0.10 &
Failure rates 0.001, 0.04 &
0.0 \\
Filter &
Failure rates 0.04, 0.12 &
Failure rates 0.02, 0.08 &
Failure rates 0.001, 0.03 &
0.0 \\
\bottomrule
\end{tabularx}
\label{tab:candidate_pool_main}
\end{table*}

\begin{table*}[ht]
\caption{Default simulated network topology. The simulator interprets each entry as a logical inter-zone link with stochastic delay, bandwidth, and message-level loss.}
\centering
\footnotesize
\setlength{\tabcolsep}{5pt}
\renewcommand{\arraystretch}{1.08}
\begin{tabular}{lcccc}
\toprule
\textbf{Link} & \textbf{Base latency (ms)} & \textbf{Latency Jitter (ms)} &
\textbf{Bandwidth (tokens/ms)} & \textbf{Loss rate} \\
\midrule
Local--Local & 2 & 1 & 500 & 0.0001 \\
Local--Regional & 150 & 30 & 50 & 0.015 \\
Local--Global & 500 & 150 & 10 & 0.03 \\
Regional--Regional & 40 & 10 & 100 & 0.005 \\
Regional--Global & 350 & 60 & 20 & 0.008 \\
Global--Global & 60 & 15 & 200 & 0.001 \\
\bottomrule
\end{tabular}
\label{tab:candidate_pool_default_network}
\end{table*}

For the main comparison, Random and BestOne select a single candidate from this
pool, MoA executes the full pool for the invoked role, and HACO selects a hedge
subset according to its reliability target. The fixed-zone baselines use
restricted variants of the same role structure. The Local baseline keeps one
Local Llama candidate per role with failure rate 0.001. The Regional baseline
keeps one Regional \texttt{gpt-4o-mini} candidate per role with failure rate
0.02. The Global baseline keeps one Global \texttt{gpt-4o} candidate per role
with failure rate 0.05. Therefore, each fixed-zone baseline uses four role
candidates in total, all from a single execution zone. These fixed-zone
configurations use the same default simulated network topology as
Table~\ref{tab:candidate_pool_default_network}.

\subsection{Real-World Azure Setting Details.}
\label{app:realworld_setting}

This section describes the real-world deployment used for the Azure-based HACO
experiments. Unlike the simulated experiments in
Appendix~\ref{app:network_simulation}, this setting does not sample latency from
a hand-specified topology. Instead, each candidate agent is backed by an Azure
OpenAI deployment in a concrete cloud region, and the runtime measures endpoint
round-trip time during execution.

\paragraph{Azure candidate pool.}
The deployment uses the same four MAS roles as the main experiments: planner,
coder, debugger, and filter. For each role, we instantiate one candidate in each
of four Azure regions, yielding 16 role candidates in total. The router/evaluator
is a separate Japan East deployment of \texttt{gpt-4.1-mini} and is not counted
as a role candidate. In real mode, the configuration loader canonicalizes the
Azure connection fields and dispatches these entries through the
\texttt{azure-openai-real-chat} backend. Table~\ref{tab:realworld_azure_pool}
summarizes the active model and region assignments; API keys and endpoint
credentials are omitted from the paper.

\begin{table*}[ht]
\caption{Azure-backed candidate pool used in the real-world bad-agent setting.
The Local/Regional/Global labels are experimental execution labels; the actual
network path is measured against the concrete Azure endpoint for each candidate.}
\centering
\footnotesize
\setlength{\tabcolsep}{5pt}
\renewcommand{\arraystretch}{1.08}
\resizebox{\textwidth}{!}{
\begin{tabularx}{\textwidth}{
>{\raggedright\arraybackslash}p{2.2cm}
>{\raggedright\arraybackslash}p{2.6cm}
>{\raggedright\arraybackslash}p{2.8cm}
>{\raggedright\arraybackslash}X
>{\raggedright\arraybackslash}p{2.5cm}}
\toprule
\textbf{Execution label} & \textbf{Azure region} & \textbf{Model} &
\textbf{Roles instantiated} & \textbf{Forced outage in stress config} \\
\midrule
Global & East US & \texttt{gpt-5.1-chat} &
Planner, coder, debugger, filter &
Planner, coder, and filter use \texttt{real\_failure\_rate}=1.0 \\
Global & Sweden Central & \texttt{gpt-4o} &
Planner, coder, debugger, filter &
Debugger uses \texttt{real\_failure\_rate}=1.0 \\
Regional & Southeast Asia & \texttt{gpt-4.1-nano} &
Planner, coder, debugger, filter &
None \\
Local & Japan East & \texttt{gpt-4.1-mini} &
Planner, coder, debugger, filter; router/evaluator &
None \\
\bottomrule
\end{tabularx}}
\label{tab:realworld_azure_pool}
\end{table*}

All deployments use Azure API version \texttt{2024-12-01-preview},
\texttt{top\_p}=1, and frequency penalty 0. The planner, debugger, filter, and
router/evaluator are run with temperature 0.0, while the coder role uses
temperature 0.2. The bad-agent configuration injects failures only through
\texttt{real\_failure\_rate} in real mode; candidates with value 1.0 fail before
the model request is issued and are logged as agent failures, while candidates
with value 0 remain available unless the actual Azure request fails.

\paragraph{Real network-latency measurement.}

For each real candidate invocation, the directly observed API latency is an
end-to-end latency: it includes client serialization, Internet and cloud routing,
Azure front-end handling, service queueing, model inference, response transfer,
and client-side response processing. Azure does not expose a per-request
decomposition that cleanly separates pure network delay from model-side
processing time. We therefore log this observed full-call latency as the real
candidate execution latency, while collecting an additional RTT signal to build
the network profile used by HACO.

Concretely, the runtime performs a lightweight sidecar probe to the target
candidate's Azure endpoint by issuing HTTP requests to the Azure
deployment-listing endpoint for that resource. For each probe, it measures
wall-clock round-trip time with a monotonic timer. The default probe
configuration uses three samples, a 5 second per-request timeout, a 30 second
cache TTL per endpoint and API version, and a 0.8 second fallback latency if
every probe attempt fails. Any HTTP response is sufficient for measuring RTT,
because the probe is used only to characterize transport reachability and
endpoint responsiveness, not generation quality.

The median successful probe time is recorded as the endpoint RTT observation for
that candidate's network profile. The runtime also serializes the transmitted
notebook cells to JSON and records the UTF-8 byte size as the payload size. In
real-probe mode, the source and target execution-zone labels are retained for
reporting, but the network profile is endpoint-based: it measures the
orchestrator-to-Azure round-trip time to the target candidate's region, rather
than a synthetic zone-to-zone delay. The experience log aggregates these RTT
observations into average RTT, RTT jitter, and probe failure rate for each
observed agent pair. During HACO selection, the router uses the average RTT when
available, or performs a fresh endpoint probe for cold-start candidates. This
RTT profile serves as the real-deployment analogue of the communication-latency
term in the simulated setting: it is combined with the candidate's historical
execution latency and then passed through the same bounded logarithmic latency
discount used in the main HACO objective.

\subsection{Compute Resources.}
\label{app:compute_res}
All experiments were executed in Dockerized environments (Docker 29.1.3) on a Linux server with AMD EPYC 9J14 CPUs (2 sockets, 96 cores per socket, 384 logical CPUs), 755\,GiB system memory, and NVMe storage (6.9\,TB total). 
Our method does not require local GPU training/inference; all model inference is performed via external LLM APIs, and the same API/deployment settings are used across compared methods for fairness. 
Although the host machine has GPUs available, they were not used in the reported experiments. 
Detailed runtime/token/cost statistics are already reported in Appendix Table~\ref{tab:main_results}.

\section{Algorithms and Baselines}
\subsection{HACO Algorithm}
\label{sec:algorithm}

Algorithm~\ref{alg:haco} summarizes HACO.
Here, $\tau$ specifies the target system reliability, $\lambda$ controls the exploration bonus in optimistic ranking, and $\gamma$ controls the conservative margin used in hedge-set stopping.
The controller mapping $\Phi$ takes the current query-context pair $(q,\mathcal{C}_t)$ and returns a role invocation event $(r_t,x_t)$.
The initial task context $\mathcal{C}_0$ is initialized from the original query $q$, input data, and an empty execution history.

\begin{algorithm}[t]
\small
\caption{HACO: Hedged Agent Computing}
\label{alg:haco}

\SetKwInOut{Input}{Input}
\SetKwInOut{Output}{Output}

\Input{Query $q$, target system reliability $\tau$, exploration $\lambda$, risk $\gamma$}
\Output{Final response}

Initialize task context $\mathcal{C}_0$\tcp*[r]{Initialize context}

\While{task not completed}{

    \phase{Phase 1: Role Abstraction}
    
    $(r_t, x_t) \gets \Phi(q, \mathcal{C}_t)$\tcp*[r]{Infer role invocation event}
    $\mathcal{A}_t \gets$ candidate agent instances for $r_t$\tcp*[r]{Retrieve candidates by role type}

    \phase{Phase 2: Redundancy Allocation}

    \ForEach{$a_i \in \mathcal{A}_t$}{
        $(\mu_i,\sigma_i) \gets \text{PosteriorStats}(\alpha_i,\beta_i)$\tcp*[r]{Posterior stats}
        $\rho_i \gets (1-f_{\mathrm{int}}^{(i)})(1-\ell_{\mathrm{net}}^{(i)})$\tcp*[r]{Execution success probability}
        $C_i \gets \ln\!\left(1+\frac{T_i}{T_0}\right)$\tcp*[r]{latency-aware resource cost}
        $D_i \gets 1 + \eta C_i$\tcp*[r]{Stable delay discount}
        $U_i^{+} \gets \frac{\min(1,\mu_i+\lambda\sigma_i)\rho_i}{D_i}$\tcp*[r]{Optimistic ranking}
    }

    Sort $\mathcal{A}_t$ by $U_i^{+}$ descending\tcp*[r]{Rank by optimistic utility}

    $\mathcal{S}_t \gets \emptyset$, $P_{\mathrm{fail}} \gets 1$\tcp*[r]{Initialize hedge}

    \ForEach{$a_i$ in sorted $\mathcal{A}_t$}{
        $R^{-}_{i} \leftarrow \min(0.98, \max(\epsilon_{\rm R}, \mu_i - \gamma\sigma_i)\cdot\rho_i)$\tcp*[r]{Conservative execution reliability}
        $\mathcal{S}_t \gets \mathcal{S}_t \cup \{a_i\}$\tcp*[r]{Add to hedge}
        $P_{\mathrm{fail}} \gets P_{\mathrm{fail}}(1-R_i^{-})$\tcp*[r]{Update failure}
        \If{$1-P_{\mathrm{fail}} \ge \tau$}{
            break\tcp*[r]{Meet target system reliability}
        }
    }

    \phase{Phase 3: Hedged Execution}

    Launch $\mathcal{S}_t$ on payload $x_t$ in parallel\tcp*[r]{Parallel launch}
    Monitor completions until $a_{\mathrm{winner}}$ is determined or all candidates fail\tcp*[r]{Decision-time latency}
    $\mathcal{S}_t^{+} \gets \{a_i \in \mathcal{S}_t \mid s_i = 1\}$\tcp*[r]{Successful candidates}
    \If{$\mathcal{S}_t^{+} = \emptyset$}{
        $\mathcal{C}_{t+1} \gets \mathrm{IntegrateFailure}(\mathcal{S}_t)$\tcp*[r]{Trigger recovery}
        \textbf{continue}
    }
    $a_{\mathrm{winner}} \gets \arg\max_{a_i \in \mathcal{S}_t^{+}} U_i^{+}$\tcp*[r]{Winner selection}
    $\mathcal{C}_{t+1} \gets \mathrm{Integrate}(a_{\mathrm{winner}})$\tcp*[r]{Context update}
    Collect remaining traces asynchronously\tcp*[r]{Experience harvesting}
    
    \phase{Phase 4: Experience Harvesting}
    
    \ForEach{$a_i \in \mathcal{S}_t$}{
        Record execution trace\tcp*[r]{$(q_i,s_i,t_i)$ + telemetry}
        Update candidate profile\tcp*[r]{Candidate statistics}
        Update link profile\tcp*[r]{Link statistics}
    }
}
\Return final response

\end{algorithm}

\paragraph{Hyperparameters.}
\label{app:hyperparameters}
Unless otherwise specified, all HACO results use 
$\tau=0.9$, $\lambda=1.1$, $\gamma=0.7$, $\kappa=2.0$, 
$\eta=0.5$, $\epsilon_{\rm B}=10^{-6}$, and $\epsilon_{\rm R}=0.1$.
Here, $\tau$ is the target system reliability, $\lambda$ is the optimistic-ranking coefficient, 
$\gamma$ is the conservative-margin coefficient, $\kappa$ is the bandwidth-jitter penalty, 
$\eta$ is the latency-discount strength, $\epsilon_{\rm B}$ is the bandwidth denominator floor, 
and $\epsilon_{\rm R}$ is the conservative reliability lower bound. 
Only $\tau$ is varied in the reliability-target sweep. For the BestOne baseline, we use cold-start threshold $n_{\min}=2$ and
epsilon-greedy exploration probability $p_{\rm exp}=0.1$.

\subsection{Properties of HACO's Redundancy Allocation}
\label{app:haco_properties}
\paragraph{Scope of the guarantee.}
The results in this section should be interpreted as conditional properties of the HACO allocation rule, with their validity depending on explicit modeling assumptions. In particular, the conservative success statement relies on two conditions: the candidate qualified-success estimates \(R_i^-\) are calibrated lower bounds on the corresponding true probabilities, and selected candidates are conditionally independent given the observed runtime state. Under these assumptions, HACO's accumulation rule provides a conservative certificate for the selected hedge set. When the assumptions are violated, the certificate should be viewed as an approximate reliability estimate, and empirical monitoring or recalibration is required.

We provide several basic properties of HACO's Phase-2 redundancy allocation rule.
Recall that candidates are sorted by descending optimism-adjusted utility
\(
U_i^+
\),
and then added sequentially using the conservative execution reliability estimate
\(
R_i^- \in [0, 0.98]
\),
with cumulative failure probability updated as
\(
P_{\mathrm{fail}} \leftarrow P_{\mathrm{fail}}(1-R_i^-)
\).
The algorithm stops once
\(
1-P_{\mathrm{fail}} \ge \tau
\).

\paragraph{Notation.}
Let the sorted candidate list be
\(
a_{(1)},a_{(2)},\dots,a_{(m)}
\),
where
\(
U_{(1)}^+ \ge U_{(2)}^+ \ge \cdots \ge U_{(m)}^+
\).
For the first \(k\) candidates in this order, define
\[
P_{\mathrm{fail}}^{(k)} := \prod_{j=1}^{k}(1-R_{(j)}^-),
\qquad
P_{\mathrm{succ}}^{(k)} := 1 - P_{\mathrm{fail}}^{(k)}.
\]
By convention,
\(
P_{\mathrm{fail}}^{(0)}=1
\)
and
\(
P_{\mathrm{succ}}^{(0)}=0
\).

\paragraph{Reachability assumption.}
Unless otherwise stated, the following guarantees are stated for the feasible
case in which the target reliability is conservatively reachable by the available
candidate pool:
\[
P_{\mathrm{succ}}^{(m)}
=
1-\prod_{j=1}^{m}(1-R^-_{(j)})
\ge \tau .
\]
If this condition does not hold, HACO returns the full candidate pool under the
given ordering. In that case, the conservative success estimate is
\(P_{\mathrm{succ}}^{(m)}\), and no method restricted to the same candidate pool
and conservative estimates can certify the target \(\tau\).

\begin{lemma}[Monotonicity of conservative accumulation]
\label{lem:monotonicity}
Assume \(R_{(j)}^- \in [0,1]\) for all \(j\).
Then \(P_{\mathrm{fail}}^{(k)}\) is nonincreasing in \(k\), and
\(P_{\mathrm{succ}}^{(k)}\) is nondecreasing in \(k\).
\end{lemma}

\begin{proof}
Since \(R_{(k+1)}^- \in [0,1]\), we have \(1-R_{(k+1)}^- \in [0,1]\). Therefore
\[
P_{\mathrm{fail}}^{(k+1)}
=
P_{\mathrm{fail}}^{(k)}(1-R_{(k+1)}^-)
\le
P_{\mathrm{fail}}^{(k)}.
\]
Hence \(P_{\mathrm{fail}}^{(k)}\) is nonincreasing, and thus
\(
P_{\mathrm{succ}}^{(k)} = 1-P_{\mathrm{fail}}^{(k)}
\)
is nondecreasing.
\end{proof}

\begin{proposition}[Conditional conservative qualified-success certificate]
\label{prop:safe_termination}
Assume conditional independence across selected candidates, suppose that
\(R^-_i\) is a lower bound on the true qualified-success probability of candidate \(a_i\),
and assume that the target reliability is conservatively reachable, i.e.,
\(P_{\mathrm{succ}}^{(m)} \ge \tau\). If HACO terminates with hedge set \(S\),
then
\[
P(\text{qualified success}\mid S) \ge \tau .
\]
\end{proposition}

\begin{proof}
Under conditional independence, the true probability that no selected candidate achieves qualified success is
\[
\prod_{a_i \in \mathcal{S}} (1-p_i),
\]
where \(p_i\) denotes the true qualified-success probability of candidate \(a_i\).
Since \(R_i^- \le p_i\), we have
\[
1-p_i \le 1-R_i^-,
\]
and therefore
\[
\prod_{a_i \in \mathcal{S}}(1-p_i)
\le
\prod_{a_i \in \mathcal{S}}(1-R_i^-)
=
P_{\mathrm{fail}}.
\]
Thus,
\[
\mathbb{P}(\text{qualified success} \mid \mathcal{S})
=
1-\prod_{a_i \in \mathcal{S}}(1-p_i)
\ge
1-P_{\mathrm{fail}}.
\]
By the reachability assumption and the stopping rule, HACO stops at some
\(k^\star \le m\) satisfying \(1-P_{\mathrm{fail}}\ge \tau\). Therefore the result follows.
\end{proof}

\begin{proposition}[Minimal feasible prefix under HACO ordering]
\label{prop:min_prefix}
Fix the ordering induced by descending \(U_i^+\), and assume
\(P_{\mathrm{succ}}^{(m)} \ge \tau\). Suppose HACO stops after selecting the first
\(k^\star\) candidates. Then the returned hedge set is exactly the shortest
prefix of this ordered list whose conservative success estimate reaches the
target:
\[
P_{\mathrm{succ}}^{(k^\star)} \ge \tau,\qquad
P_{\mathrm{succ}}^{(k)} < \tau \quad \text{for all } k<k^\star .
\]
\end{proposition}

\begin{proof}
By construction, HACO scans the ordered list sequentially and stops at the first index
\(k^\star\) such that
\(
P_{\mathrm{succ}}^{(k^\star)} \ge \tau
\).
Therefore all earlier prefixes fail to meet the threshold, i.e.,
\(
P_{\mathrm{succ}}^{(k)} < \tau
\)
for \(k < k^\star\).
Hence the selected set is the shortest feasible prefix under the HACO ordering.
\end{proof}

\begin{proposition}[Monotonicity with respect to the system reliability target]
\label{prop:tau_monotone}
Fix the ordered candidate list and the conservative execution reliability estimates
\(\{R_{(j)}^-\}_{j=1}^m\), and assume
\(P_{\mathrm{succ}}^{(m)} \ge \tau_2\).
Let \(k^\star(\tau)\) denote the stopping index of HACO under target system reliability \(\tau\).
Then \(k^\star(\tau)\) is nondecreasing in \(\tau\):
if \(\tau_1 \le \tau_2\), then
\[
k^\star(\tau_1) \le k^\star(\tau_2).
\]
\end{proposition}

\begin{proof}
From Lemma~\ref{lem:monotonicity}, \(P_{\mathrm{succ}}^{(k)}\) is nondecreasing in \(k\).
For a lower threshold \(\tau_1\), HACO stops at the smallest prefix whose system reliability estimate
reaches \(\tau_1\). For a higher threshold \(\tau_2 \ge \tau_1\), any prefix feasible for
\(\tau_2\) is also feasible for \(\tau_1\), but not necessarily conversely. Therefore the minimal feasible
prefix for \(\tau_2\) cannot be shorter than that for \(\tau_1\).
\end{proof}

\begin{proposition}[Single-step overshoot bound]
\label{prop:overshoot}
Assume \(P_{\mathrm{succ}}^{(m)} \ge \tau\), and let \(k^\star\) be the stopping
index. Then HACO's conservative success estimate satisfies
\[
0 \le P_{\mathrm{succ}}^{(k^\star)}-\tau
\le
P_{\mathrm{fail}}^{(k^\star-1)} R^-_{(k^\star)} .
\]
\end{proposition}

\begin{proof}
We have
\[
P_{\mathrm{succ}}^{(k^\star)}
-
P_{\mathrm{succ}}^{(k^\star-1)}
=
P_{\mathrm{fail}}^{(k^\star-1)}R_{(k^\star)}^-.
\]
Since \(k^\star\) is the first index reaching the threshold,
\[
P_{\mathrm{succ}}^{(k^\star-1)} < \tau \le P_{\mathrm{succ}}^{(k^\star)}.
\]
Hence
\[
0 \le P_{\mathrm{succ}}^{(k^\star)}-\tau
\le
P_{\mathrm{succ}}^{(k^\star)}-P_{\mathrm{succ}}^{(k^\star-1)}
=
P_{\mathrm{fail}}^{(k^\star-1)}R_{(k^\star)}^-.
\]
\end{proof}

\subsection{Baselines}
\label{app:baseline}

We compare our method with three routing baselines, denoted as \textsc{Random}, \textsc{BestOne}, and \textsc{MoA-style Best-of-N} in the main text.

\paragraph{\textsc{Random}.}
This baseline uniformly samples one candidate from the candidate pool and executes only that candidate.
Its output is returned directly without any further reranking or model-based evaluation.

\begin{algorithm}[t]
\caption{\textsc{Random}}
\label{alg:random-router}
\KwIn{Query $x$, candidate set $\mathcal{A}$}
\KwOut{Final response $y$}
Uniformly sample one candidate $a$ from $\mathcal{A}$\;
Execute candidate $a$ on query $x$ and obtain response $y$\;
\Return{$y$}\;
\end{algorithm}

\paragraph{\textsc{BestOne}.}
This baseline performs single-candidate selection based on historical performance.
It uses the same LLM-judge quality signal as HACO for executed candidates, with failed or invalid executions assigned zero utility.
It prefers the candidate with the best past average utility, while retaining a small probability of selecting other candidates for exploration.
Only the selected candidate is executed and updated, and its response is returned as the final output.

\begin{algorithm}[t]
\caption{\textsc{BestOne}}
\label{alg:bestone-router}
\KwIn{Query $x$, candidate set $\mathcal{A}$, exploration probability $p_{\rm exp}$}
\KwOut{Final response $y$}
\ForEach{candidate $a \in \mathcal{A}$}{
    Estimate its historical average utility $u(a)$\;
    \If{candidate $a$ has no history}{
        Assign an optimistic default utility\;
    }
}
Let $a^\star$ be the candidate with the highest estimated utility\;
Draw a random number $r$ from $[0,1]$\;
\eIf{$r > p_{\rm exp}$}{
    Select $a^\star$\;
}{
    Randomly select one candidate from $\mathcal{A} \setminus \{a^\star\}$\;
}
Execute the selected candidate on query $x$ and obtain response $y$\;
\Return{$y$}\;
\end{algorithm}

\begin{algorithm}[t]
\caption{\textsc{Algorithm 4: MoA-style Best-of-N Router}}
\label{alg:moa-router}
\KwIn{Query $x$, candidate set $\mathcal{A}$}
\KwOut{Final response $y$}
Execute all candidates in $\mathcal{A}$ on query $x$ in parallel\;
Collect all candidate responses $\{y_a \mid a \in \mathcal{A}\}$\;
Use an evaluator language model to assign an output-quality score to each candidate response\;
Select the highest-scoring response as $y$\;
\Return{$y$}\;
\end{algorithm}

\paragraph{MoA-style Best-of-N.}
This baseline executes all candidates in parallel and collects all generated responses. A separate evaluator language model then assigns output-quality scores to the candidate outputs, and the highest-scoring response is selected as the final answer. We use ``MoA'' as a compact figure and table label for this MoA-style full-pool evaluator-based baseline.

\paragraph{Discussion.}
These baselines represent three distinct routing paradigms: random single-candidate routing, experience-guided single-candidate routing with limited exploration, and parallel candidate generation followed by model-based selection.
They provide simple reference points for evaluating the effectiveness-efficiency trade-off of routing strategies.

\section{More Experimental Details}
\label{app:main}

\subsection{Radar Axis Definitions}
\label{app:radar_metric_defs}
For each benchmark $b$ and method $a$, the radar in Fig.~\ref{fig:radar} is built from five raw metrics and then mapped to a common higher-is-better radial scale.

\paragraph{Effectiveness and strict success.}
Let $\mathcal{T}_{a,b}$ be tasks of method $a$ on benchmark $b$.
The first axis (``Avg. Score'' or ``Accuracy'') is the benchmark-specific mean effectiveness in percentage:
\[
M^{\text{eff}}_{a,b} =
\begin{cases}
\frac{1}{|\mathcal{T}_{a,b}|}\sum_{t\in\mathcal{T}_{a,b}} \text{score}_t, & b\in\{\text{MatplotBench},\text{DSBench}\},\\[3pt]
\frac{1}{|\mathcal{T}_{a,b}|}\sum_{t\in\mathcal{T}_{a,b}} \text{acc}_t, & b=\text{InfiAgent-Bench}.
\end{cases}
\]
The strict-success axis is:
\[
M^{\text{strict}}_{a,b} =
\begin{cases}
100\cdot\frac{1}{|\mathcal{T}_{a,b}|}\sum_{t}\mathbf{1}[\text{score}_t\ge 80], & b=\text{MatplotBench},\\[3pt]
100\cdot\text{question-level accuracy}, & b=\text{InfiAgent-Bench},\\[3pt]
100\cdot\frac{\#\{t:\text{score}_t>0\}}{|\mathcal{T}_{a,b}|}, & b=\text{DSBench}.
\end{cases}
\]

\paragraph{Raw cost and failure metrics.}
From end-of-run traces, we compute mean token usage and latency; LLM-judge tokens for $q_i$ are counted in cost metrics, while asynchronous judging does not add to blocking latency:
\[
M^{\text{tok}}_{a,b}=\frac{1}{1000}\cdot \frac{1}{|\mathcal{T}_{a,b}|}\sum_t \text{token\_count}_t,\qquad
M^{\text{lat}}_{a,b}=\frac{1}{|\mathcal{T}_{a,b}|}\sum_t \text{latency}_t.
\]
For reliability, let $\mathcal{E}_{a,b}$ be all observed task steps and let $\mathcal{B}_{t,s}$ be executed concurrent branches at step $s$ of task $t$.
With branch success indicator $z_{t,s,k}\in\{0,1\}$, the all-failed-step rate is
\[
M^{\text{fail}}_{a,b}
=100\cdot\frac{1}{|\mathcal{E}_{a,b}|}
\sum_{(t,s)\in\mathcal{T}_{a,b}}
\mathbf{1}\!\left[\sum_{k\in\mathcal{B}_{t,s}} z_{t,s,k}=0\right].
\]
Lower $M^{\text{fail}}_{a,b}$ means higher reliability.

\paragraph{``Token Eff.'', ``Latency Eff.'', and ``Reliability'' in the radar.}
For each benchmark and each axis $j$, let $x_{a,b,j}$ be the raw value above, and define method-wise extrema
$x^{\min}_{b,j}=\min_a x_{a,b,j}$, $x^{\max}_{b,j}=\max_a x_{a,b,j}$.
We then apply the same monotonic mapping as in plotting:
\[
\hat{x}_{a,b,j}=
\begin{cases}
\dfrac{x_{a,b,j}-x^{\min}_{b,j}}{x^{\max}_{b,j}-x^{\min}_{b,j}}, & \text{higher-is-better axes},\\[8pt]
\dfrac{x^{\max}_{b,j}-x_{a,b,j}}{x^{\max}_{b,j}-x^{\min}_{b,j}}, & \text{lower-is-better axes}.
\end{cases}
\]
Thus, ``Token Eff.'' and ``Latency Eff.'' are reversed versions of mean tokens and mean latency, and ``Reliability'' is the reversed all-failed-step rate. The implementation adds a small per-axis padding before radius projection for visual spacing, but this does not change the ordering of methods.

\paragraph{Raw overall benchmark results.}
Table~\ref{tab:main_results} reports the raw numerical results used to construct
the radar comparison in Fig.~\ref{fig:radar}. We report two effectiveness metrics for each benchmark: Avg. Score / Acc. and Strict Succ. (\%).
Because the benchmarks use different evaluation protocols, the exact definitions are benchmark-specific.

For MatplotBench, Avg. Score / Acc. is the mean evaluator score over the plotting tasks. Strict Succ. (\%) is SR@80, which counts a task as successful if its evaluator score is at least 80.

For InfiAgent-Bench, Avg. Score / Acc. is the proportional sub-question accuracy using all benchmark questions as the denominator. Each question receives partial credit according to the fraction of its sub-questions answered correctly, while unevaluated questions are counted as zero. Strict Succ. (\%) is the question-level success rate using all benchmark questions as the denominator; a question is counted as successful only when all of its sub-questions are answered correctly.

For DSBench, Avg. Score / Acc. is the average normalized task score over all tasks, where task performance is normalized relative to the baseline and ground-truth performance. Tasks with missing or invalid outputs are counted as zero. Strict Succ. (\%) is the percentage of tasks with a positive normalized score, meaning the method improves over the baseline on that task.

These benchmark-specific definitions also explain why the relative magnitude of
Avg. Score / Acc. and Strict Succ. (\%) differs across benchmarks. On InfiAgent-Bench,
partial sub-question credit can make Avg. Score / Acc. higher than Strict Succ. On
DSBench, many tasks may obtain positive but small normalized scores, so Strict Succ.
can be higher than Avg. Score / Acc. On MatplotBench, the relation depends on the
distribution of evaluator scores around the SR@80 threshold. Therefore, the two
columns should be interpreted according to each benchmark's metric definition rather
than compared using a universal ordering.

\subsection{Note on MoA-style Baseline.}
MoA-style baseline's final score depends on post-hoc evaluator selection over heterogeneous candidate outputs. 
On InfiAgent-Bench, tasks are relatively easy and candidate outputs are often close, so full-pool execution brings limited marginal gain. 
On DSBench, tasks are more difficult and the pool may not contain a clearly superior output. 
On MatplotBench, diverse visualization programs and layouts further increase evaluator-selection noise, so MoA improves step-level availability but does not necessarily achieve the best task-level score.

\begin{table*}[t]
    \centering
    \small
    \caption{Detailed information about overall benchmark performance results. Effectiveness metrics are dataset-specific. ``Strict Succ.'' denotes SR@80 on MatplotBench and question-level accuracy on InfiAgent-Bench. ``All-Failed Step Rate'' is computed as the proportion of observed task steps for which all executed concurrent branches fail.}
    \label{tab:main_results}
    \resizebox{\textwidth}{!}{
    \begin{tabular}{llccccc}
        \toprule
        \multirow{2}{*}{\textbf{Dataset}} & \multirow{2}{*}{\textbf{Method}} &
        \multicolumn{2}{c}{\textbf{Effectiveness}} &
        \multicolumn{2}{c}{\textbf{Efficiency}} &
        \multicolumn{1}{c}{\textbf{Reliability}} \\
        
        \cmidrule(lr){3-4} \cmidrule(lr){5-6} \cmidrule(lr){7-7}
        
        & & Avg. Score / Acc. & Strict Succ. (\%) & Tokens (k) & Latency (s) & All-Failed Step Rate (\%) \\
        \midrule
        
        \multirow{7}{*}{MatplotBench}

            & random       & 46.45 & 27.00 & 46.51 & 82.69 & 8.78 \\
            & bestone      & 47.75 & 25.00 & 32.22 & 52.50 & 9.35 \\
            & moa          & 42.75 & 23.00 & 305.67 & 206.41 & 0.00 \\
            & fix-local    & 23.90 & 11.00 & 43.55 & 54.42 & 0.28 \\
            & fix-regional & 47.68 & 25.00 & 64.91 & 50.74 & 3.01 \\
            & fix-global   & 52.61 & 36.00 & 16.29 & 28.62 & 8.25 \\
            & haco         & 58.86 & 36.00 & 65.32  & 23.50 & 0.17 \\
            
        \midrule

        \multirow{4}{*}{InfiAgent-Bench}

            & random       & 86.06 & 83.27 & 25.82 & 34.87 & 8.65 \\
            & bestone      & 77.82 & 75.49 & 25.33 & 22.69 & 13.94 \\
            & moa          & 84.79 & 80.93 & 204.02 & 99.82 & 0.00 \\
            & fix-local    & 69.42 & 64.98 & 52.37 & 132.47 & 0.13 \\
            & fix-regional & 87.86 & 84.44 & 35.81 & 25.82 & 3.56 \\
            & fix-global   & 4.28 & 4.28 & 14.27 & 13.23 & 8.59 \\
            & haco         & 89.01 & 84.44 & 91.04 & 31.31 & 0.37 \\

        \bottomrule

        \multirow{4}{*}{DSBench}

            & random       & 31.83 & 58.90 & 89.54 & 254.04 & 6.95 \\
            & bestone      & 41.84 & 73.61 & 98.90 & 107.41 & 11.14 \\
            & moa          & 41.09 & 71.23 & 711.73 & 423.41 & 0.11 \\
            & fix-local    & 4.90 & 10.96 & 103.52 & 546.77 & 0.00 \\
            & fix-regional & 39.79 & 67.12 & 88.39 & 143.46 & 3.53 \\
            & fix-global   & 33.02 & 61.64 & 60.28 & 75.21 & 7.64 \\
            & haco         & 42.40 & 76.71 & 393.72 & 125.91 & 0.40 \\
            
        \bottomrule

    \end{tabular}}
\end{table*}

Fig.~\ref{fig:analysis} further analyzes these raw results by comparing the
distribution of latency and token usage, together with task effectiveness and
step-level failure rate, across all baselines and HACO.
\begin{figure}[ht]
    \centering
    \includegraphics[width=1.0\textwidth]{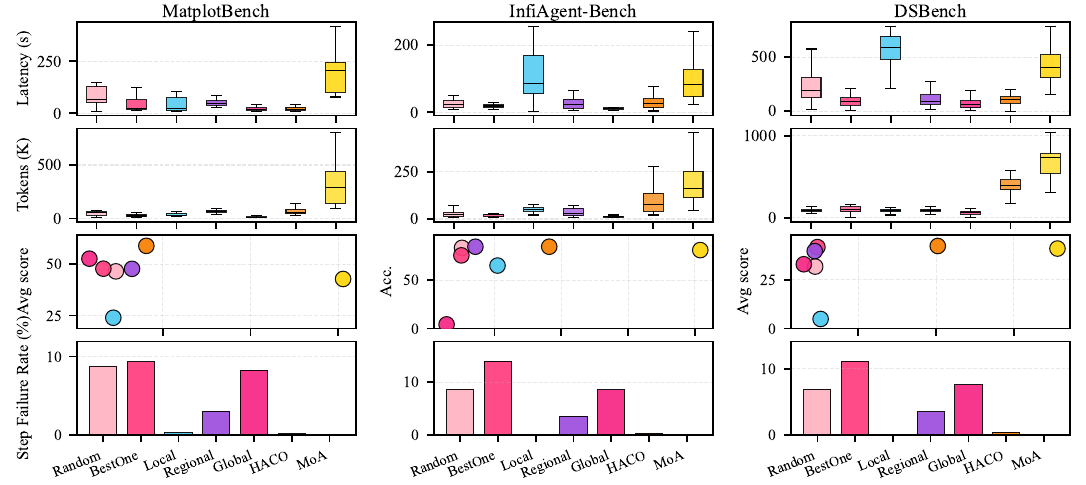}
    \caption{\textbf{More analysis of raw benchmark results.}
    The figure compares latency, token usage, task effectiveness, and step-level
    failure rate for baselines and HACO across MatplotBench, InfiAgent-Bench,
    and DSBench.}
    \label{fig:analysis}
\end{figure}

\subsection{Note on fixed-global in InfiAgent-Bench.}
\label{app:note}
The low InfiAgent-Bench score of the fixed-global baseline mainly reflects its low valid-output completion rate, while completed outputs remain of reasonable quality.
When fixed-global produces a valid final output, the judge usually assigns a full score, but many runs fail to produce a valid evaluable answer and are counted as zero in the aggregate accuracy. 
The reported all-failed-step rate is a step-level execution-availability metric and does not directly measure task-level valid-output completion. 
Thus, the anomalously low average score reflects completion failures under the fixed global route, not a systematic quality drop of the underlying GPT-4o model.

\subsection{Execution Environment-Adjusted Reference to External Baselines on MatplotBench}
\label{app:external_sota_estimate}

To provide an external reference, we compare HACO with strong results reported in the original DATAWISE paper~\cite{2025datawiseagent} based on MatplotBench. 
Those results were obtained under the original benchmark setting, which does not model the practical heterogeneous runtime conditions considered in this work. We therefore report them as \textit{idealized execution environment-free} scores and provide an \textit{execution environment-adjusted reference} under our execution setting.

The adjustment is calibrated using two measured anchor points. The first anchor is the DATAWISE w/ visual tool score reported in~\cite{2025datawiseagent}. The second anchor is the best fixed-zone MatplotBench score measured in our heterogeneous execution setting, where the same type of agent workflow is affected by the runtime conditions used in our experiments. 
We apply this calibration to the external methods to approximate the performance decrease caused by our network-aware execution setting. The HACO row is directly observed in our experiments. Thus, 
Table~\ref{tab:external_sota_matplot_adjusted} provides a contextual comparison that approximates how real-world execution environment factors may affect results originally reported under idealized execution environment-free benchmark settings.

\begin{table}[t]
\centering
\small
\caption{
External reference on MatPlotBench under idealized execution environment-free and
execution environment-adjusted settings. 
The idealized scores are taken from DATAWISE/DatawiseAgent~\cite{2025datawiseagent};
the adjusted scores are estimated from those reported results under our heterogeneous
execution-environment setting. For DATAWISE, we also report the result after adding
HACO under the same adjusted setting.
}
\label{tab:external_sota_matplot_adjusted}
\begin{tabular}{lcc}
\toprule
Method / Reference &
\shortstack{Idealized execution \\ environment-free setting} &
\shortstack{Execution \\ environment-adjusted setting} \\
\midrule
DATAWISE w/ visual tool~\cite{2025datawiseagent} &
\textbf{64.33} &
\begin{tabular}{@{}cc@{}}
52.61 & \textbf{58.86} \\
{\small (w/o HACO)} & {\small (w/ HACO)}
\end{tabular} \\
AutoGen w/ visual tool~\cite{wu2024autogen} & 63.60 & 52.0 \\
DATAWISE~\cite{2025datawiseagent} & 61.22 & 50.1 \\
MatplotAgent~\cite{yang2024matplotagent} & 57.86 & 47.3 \\
Direct Decoding~\cite{2025datawiseagent} & 45.28 & 37.0 \\
\bottomrule
\end{tabular}
\end{table}

\section{Limitations}

    \paragraph{Task scope.}
    Our validation instantiates HACO on a data-analysis-oriented agent workflow and evaluates it on benchmarks covering data science, structured data analysis, and scientific visualization. These workflows provide controlled long-horizon testbeds with code execution, failure recovery, and measurable quality signals. However, HACO is a general runtime allocation layer that can be applied to any agent-based system exposing role invocation events and candidate agent instances. Future work should validate HACO on broader workflows, such as web automation, retrieval-intensive agents, and distributed intelligent systems, where the execution-environment factors may become more pronounced because communication links among edge devices can fluctuate more strongly.
    \paragraph{Reliability-target tuning.}
    HACO treats the target reliability $\tau$ as an explicit operating parameter for the reliability--cost trade-off. In practice, $\tau$ should be tuned according to application requirements, budget, and latency constraints. Future work can further study automatic target selection and adaptive reliability scheduling across tasks with different risk levels.

    \paragraph{Feedback signals for experience harvesting.}
    The current experience harvesting module uses operational feedback such as quality, success, latency, token usage, and network statistics. Future work can incorporate richer semantic feedback, including error types, reasoning traces, code-level failure patterns, and evaluator rationales. In addition, we plan to introduce temporal adaptation mechanisms, such as forgetting factors, temporal discounting, sliding-window updates, or drift detection, so that the Beta-based capability profiles can respond more quickly to recent observations and better track non-stationary candidate behavior.
\section{Broader Impacts}
\label{app:broader_impacts}
This work studies reliability control for LLM-based agent systems under heterogeneous runtime conditions. Positive impacts are that more reliable role invocation can reduce failures in long-horizon workflows, improve service continuity under regional or network degradation, and lower unnecessary token and latency costs compared with exhaustive parallel execution. These properties can benefit practical agent deployments in domains such as data analysis, engineering assistance, and scientific workflows, where unstable execution can otherwise waste resources or interrupt user-facing services.

Potential negative impacts should also be considered. Improving execution reliability may make agent systems easier to deploy at scale, which can amplify harms already associated with LLM systems, such as broader resource consumption from increased inference activity. In addition, system-level robustness does not guarantee factual correctness, fairness, privacy preservation, or resistance to malicious use, so more reliable execution could still propagate harmful outputs if underlying models or tools fail in these dimensions.
We therefore view HACO as a systems mechanism. It could be deployed with application-level safeguards such as human oversight in high-stakes settings, usage restrictions consistent with provider policies, and monitoring of cost, latency, and misuse patterns.

\end{document}